\documentclass[a4paper,UKenglish,cleveref, autoref, thm-restate,numberwithinsect]{lipics-v2021}

\usepackage{amssymb}
\usepackage{amsmath}
\usepackage{amsthm}
\usepackage{mathrsfs}
\usepackage{mathtools}
\usepackage{bbm}
\DeclareMathOperator{\MM}{MM}

\usepackage{multirow}
\usepackage{caption}
\nolinenumbers

\usepackage{tikz}
\usepackage{pgfplots}
\pgfplotsset{compat=1.17}

\usepackage[mathscr]{eucal}
\usepackage{subcaption}

\newcommand{\bigO}{\mathcal{O}}

\usepackage{comment,color}

\newtheorem{hypothesis}[theorem]{Hypothesis}

\usepackage{mathtools}

\usepackage{cancel}
\usepackage[linesnumbered,boxed]{algorithm2e}

\SetCommentSty{mycommfont}

\newcommand{\rem}[3]{\textcolor{blue}{\textsc{#1 #2:}}
  \textcolor{red}{\textsf{#3}}}
\newcommand{\marvin}[2][says]{\rem{Marvin}{#1}{#2}}
\newcommand{\mirza}[2][says]{\rem{Mirza}{#1}{#2}}

\newcommand{\FOP}[1]{$\mathrm{FOP}_k$}

\makeatletter
\renewcommand\paragraph{%
  \@startsection{paragraph}
    {4}
    {\z@}
    {3.25ex \@plus1ex \@minus.2ex}
    {-1em}
    {\normalfont\normalsize\bfseries\addperiod}}
\newcommand{\addperiod}[1]{#1\@addpunct{.}}
\usepackage{textgreek}
\usepackage{apxproof}
\Copyright{{}} 



\category{} 

\relatedversion{}

\title{Fine-Grained Classification Of Detecting Dominating Patterns}

\author{Jonathan Dransfeld}{Karlsruhe Institute of Technology}{jonathan.dransfeld@kit.edu}{}{}
\author{Marvin Künnemann}{Karlsruhe Institute of Technology}{marvin.kuennemann@kit.edu}{https://orcid.org/0000-0003-4813-4852}{}
\author{Mirza Redzic}{Karlsruhe Institute of Technology}{mirza.redzic@kit.edu}{https://orcid.org/0009-0001-7509-1686}{}

\authorrunning{J. Dransfeld, M.Künnemann, M. Redzic} 
\funding{Research supported by the Deutsche Forschungsgemeinschaft (DFG, German Research Foundation) – 462679611.}

\EventEditors{Anne Benoit, Haim Kaplan, Sebastian Wild, and Grzegorz Herman}
\EventNoEds{4}
\EventLongTitle{33rd Annual European Symposium on Algorithms (ESA 2025)}
\EventShortTitle{ESA 2025}
\EventAcronym{ESA}
\EventYear{2025}
\EventDate{September 15--17, 2025}
\EventLocation{Warsaw, Poland}
\EventLogo{}
\SeriesVolume{351}
\ArticleNo{97}

\begin{document}
\maketitle

\begin{abstract}
We consider the following generalization of dominating sets: Let $G$ be a host graph and $P$ be a pattern graph $P$. A dominating $P$-pattern in $G$ is a subset $S$ of vertices in $G$ that (1) forms a dominating set in $G$ \emph{and} (2) induces a subgraph isomorphic to $P$. The graph theory literature studies the properties of dominating $P$-patterns for various patterns $P$, including cliques, matchings, independent sets, cycles and paths. Previous work (Kunnemann, Redzic 2024) obtains algorithms and conditional lower bounds for detecting dominating $P$-patterns particularly for $P$ being a $k$-clique, a $k$-independent set and a $k$-matching. Their results give conditionally tight lower bounds if $k$ is sufficiently large (where the bound depends the matrix multiplication exponent $\omega$). We ask: Can we obtain a classification of the fine-grained complexity for \emph{all} patterns $P$? 

Indeed, we define a graph parameter $\rho(P)$ such that if $\omega=2$, then
\[ \left(n^{\rho(P)} m^{\frac{|V(P)|-\rho(P)}{2}}\right)^{1\pm o(1)} \]
is the optimal running time assuming the Orthogonal Vectors Hypothesis, for all patterns $P$ except the triangle $K_3$. Here, the host graph $G$ has $n$ vertices and $m=\Theta(n^\alpha)$ edges, where $1\le \alpha \le 2$. 

The parameter $\rho(P)$ is closely related (but sometimes different) to a parameter $\delta(P) = \max_{S\subseteq V(P)} |S|-|N(S)|$ studied in (Alon 1981) to tightly quantify the maximum number of occurrences of induced subgraphs isomorphic to $P$. Our results stand in contrast to the lack of a full fine-grained classification of detecting an arbitrary (not necessarily \emph{dominating}) induced $P$-pattern. 

%


\end{abstract}

\section{Introduction}

Among the most intensively investigated graph problems is the dominating set problem: Given a graph $G=(V,E)$, find a (small) vertex subset $S\subseteq V$ that \emph{dominates} all vertices, i.e., for each $v\in V$, we have $v\in S$ or there is some $s\in S$ with $\{s,v\}\in E$.  In many scenarios, one might not merely want to find \emph{any} dominating set, but rather a dominating set $S$ satisfying some additional requirements. Possibilities include $S$ forming a connected subgraph, admitting a perfect matching, or more generally being connected by a prescribed topology (e.g., rings or cliques) or even being fully disconnected (i.e., forming an independent set). In this work, we study a general form of such problems: For any pattern graph $H$, the (Non-induced) Dominating $P$-Pattern problem asks to determine, given a graph $G$, whether there exists a (non-)induced copy of $P$ in $G$ that dominates all vertices of $G$.

Indeed, for various patterns $P$, the structural properties of Dominating $P$-Patterns have been well investigated, e.g., Dominating Cliques~\cite{Balliu0KO23,chepoi1998note,cozzens1990dominating,dragan1994dominating,Kohler00,kratsch1990finding,kratsch1994dominating}, Dominating Independent Sets\footnote{We remark that a Dominating Independent Set is equivalent to the well-studied notion of a Maximal Independent Set.}~\cite{cho2023tight,cho2023independent,cho2023independentdom,kuenzel2023independent,pan2023improved}, Dominating (Induced) Matchings~\cite{ haynes1998paired,southey2011characterization,studer2000induced}, Dominating Cycles~\cite{FANG2021112196,FANG202343} or Dominating Paths~\cite{FaudreeFGHJ18,FaudreeGJW17,Veldman83}.
Recently, the algorithmic complexity of detecting Dominating $P$-Patterns in general graphs has been performed by Künnemann and Redzic, focusing on cliques, independent sets and perfect matchings~\cite{KuennemannR24} (see below for further details). In this work, we set out to understand the fine-grained complexity of this problem for \emph{all} patterns $P$.

Note that the Dominating $P$-Pattern problem is the natural combination of two classic problems which are heavily studied in isolation: Dominating Set and $P$-Pattern Detection.
\begin{itemize}
    \item \emph{$k$-Dominating Set}: Dominating set is a notoriously difficult problem. When parameterized by the solution size $k$ (i.e., $|S|)$, it is the arguably most natural W[2]-complete problem and hard even to approximate within $f(k)$ factors~\cite{ChalermsookCKLM17,KarthikLM19}. The best known algorithm is based on fast matrix multiplication and solves it in time $n^{k+o(1)}$ for all sufficiently large $k$~\cite{EisenbrandG04}. This is tight in the sense that an $\bigO(n^{k-\epsilon})$-time algorithm for any $k\ge 3$ and $\epsilon >0$ would refute the $k$-Orthogonal Vectors Hypothesis ($k$-OVH) and thus the Strong Exponential Time Hypothesis~\cite{PatrascuW10} (see Section~\ref{sec:preliminaries} for details). Taking the graph sparsity (i.e., the number $m$ of edges) into account as well, a recent result~\cite{FischerKR24} gives upper and conditional lower bounds establishing a tight time complexity of $mn^{k-2\pm o(1)}$ if the matrix multiplication exponent $\omega$ is equal to $2$.  
    
    \item \emph{$P$-Pattern Detection (aka Induced $P$-Subgraph Isomorphism):} The complexity of $P$-Pattern Detection for general $P$ is sensitive to the considered pattern $P$. The probably most notable special case is $k$-Clique Detection. It is the most natural $W[1]$-complete problem and also resists good approximations~\cite{ChalermsookCKLM17,LinRSW23,KarthikK22}. Its best known algorithm solves it in time essentially $\bigO(n^{\frac{\omega}{3} k})$~\cite{NesetrilP85}\footnote{If $k$ is divisible by 3.} and is conjectured to be optimal, see, e.g.,~\cite{AbboudBVW18}. For other (classes of) patterns $P$, however, only partial results are known: While for any $k$-node pattern $P$, the problem can be reduced to $k$-Clique Detection, this approach is not necessarily optimal. E.g., all 4-node patterns except clique and independent set can be detected polynomially faster than the conjectured time for clique and independent set, see in particular~\cite{KloksKM00, EisenbrandG04, VassilevskaWWWY15, DalirrooyfardVW22}. Despite significant effort (see~\cite{KloksKM00, EisenbrandG04, VassilevskaWWWY15, DalirrooyfardVVW21, DalirrooyfardVW22} for a selection), the task of finding matching upper and conditional lower bounds appears far from completed. Let us remark that also for Non-induced $P$-Pattern Detection a fine-grained classification of all patterns $P$ remains open (see, e.g.,~\cite{Marx10, DalirrooyfardVW22}). 
\end{itemize}

We ask: \emph{How does the time complexity of Dominating $P$-Pattern relate to the complexity of $P$-Pattern Detection versus to the complexity of $k$-Dominating Set?} Can we, despite the lack of a complete classification for $P$-Pattern Detection classify the fine-grained complexity of \textbf{all} \emph{dominating} patterns? 

\medskip
Previous work appears to indicate that the complexity should be governed by the domination aspect: By adapting the conditional lower bound of Patrascu and Williams for $k$-Dominating Set~\cite{PatrascuW10} (see also~\cite{FischerKR24,KuennemannR24}) it is not too difficult to obtain a conditional lower bound of $n^{k-o(1)}$ in dense graphs for \emph{any} pattern $P$, based on $k$-OVH. However, the situation becomes much more interesting for \emph{sparse} graphs. Here, it has already been observed that the fine-grained time complexity is highly sensitive to the specific pattern $P$. Specifically, Künnemann and Redzic~\cite{KuennemannR24} give some curious insights into selected patterns: (Here, to simplify the presentation, we assume that $\omega=2$ and that $k$-OVH holds.) 
\begin{itemize}
\item If $P$ is the $k$-star (a tree with a root and $k-1$ leaves), the tight time complexity is $mn^{k-2\pm o(1)}$, i.e., even in very sparse graphs with $m=\Theta(n)$, we only save a linear factor in $n$ compared to the $n^{k\pm o(1)}$ complexity in dense graphs.
\item In contrast, if $P$ is the clique with $k\ge 5$ vertices or the independent set with $k\ge 3$ vertices, we obtain a particularly simple case with tight time complexity of $(m^{\frac{k+1}{2}}/n)^{1\pm o(1)}$. In very sparse graphs with $m=\bigO(n)$, the resulting running time of $n^{\frac{k-1}{2}\pm o(1)}$ is less than the square root of the running time of $n^{k\pm o(1)}$ in dense graphs. 
\item Finally, the substantially different pattern of a perfect matching on $k\ge 4$ vertices achieves the only slightly worse time complexity of $m^{\frac{k}{2}\pm o(1)}$.
\end{itemize}
We remark that the above results suggest that in sparse graphs, Dominating $P$-Pattern shares an additional flavor with the $P$-Pattern \emph{Enumeration} problem: It is not too difficult to obtain an algorithm whose running time is roughly bounded by the time required to list all occurrences of the pattern $P$. However, it turns out that this number alone cannot fully explain the time complexity: While for the case of perfect matchings, the time complexity coincides precisely with the maximum number of occurrences of the pattern, for others (e.g., independent sets or cliques), the time complexity is polynomially less than the maximum number of occurrences of the pattern. This begs the question: which other parameter of the pattern $P$ captures the time complexity of Dominating $P$-Pattern?  

\paragraph*{Our results}
To state the time complexity for any pattern $P$, we introduce the following graph parameter $\rho(P)$. Here for any graph $P$ and $S\subseteq V(P)$, we let $N(S)$ denote the (open) neighborhood of the vertices $S$, where $N(\emptyset)=\emptyset$.\footnote{For any graph $G$, we will always denote by $V(G)$ and $E(G)$ its set of vertices and edges, respectively.}
\begin{restatable}{definition}{defParameter}\label{def:parameter}
    Let $P$ be an arbitrary pattern graph, and denote by  $I(P) \subseteq V(P)$ the set of isolated nodes of $P$.
    Choose an \emph{independent} set $S \subseteq V(P) \setminus I(P)$ maximizing $|S| - |N(S)|$; if $S$ is not uniquely defined, take any choice maximizing $|S|$.
     We define \[
     \rho(P) \coloneqq \begin{cases} |S| - |N(S)| & \text{if } S\ne \emptyset\\ -1 & \text{if } S=\emptyset. \end{cases}
     \] and set 
     \[t_P(n, m) \coloneqq n^{\rho(P)} \cdot m^{\frac{|V(P)| - \rho(P)}{2}}. \]
\end{restatable}

These quantities turn out similar (but sometimes different) to the maximum number of occurrences of the pattern graph $P$ for patterns \emph{without isolated nodes} (i.e., $I(P) = \emptyset$).
Specifically, if $I(P) =\emptyset$ and $S\ne \emptyset$, the maximum number of induced copies of $P$ in an $n$-vertex $m$-edge graph $G$ is $\Theta(t_P(n,m))$.
Indeed, for connected graphs $G$ and patterns $P$, Alon~\cite{Alon81} defines the parameter $\delta(P) \coloneqq \max_{S\subseteq V(P)} |S|-|N(S)|$ and establishes $\Theta(n^{\delta(P)} m^{\frac{|V(P)|-\delta(P)}{2}})$ as the maximum number of induced copies of $P$ in $G$ whenever $G$ has $m=\Theta(n)$ edges (which generalizes to arbitrary $m\ge n$).\footnote{To clarify the correspondence, we remark that already Alon (see~\cite[proof of Lemma 8]{Alon81}) observed that in the definition of $\delta(P)$, one may without loss of generality let $S$ range only over independent sets: for any set $S$ maximizing $|S|-|N(S)|$, the set $S' = S\setminus (S\cap N(S))$ is an independent set with $|S'|-|N(S')|\ge |S|-|N(S)|$. Thus, whenever a nonempty set $S$ maximizes $|S|-|N(S)|$, we have $\rho(P)=\delta(P)$. Alon proves that a connected  $m$-edge graph $G$ has a maximum number of $\Theta(m^{\frac{|V(P)|+\delta(P)}{2}})$ induced copies of $P$, where $P$ is any connected pattern.
The generalization to an arbitrary trade-off between $n$ and $m$ is implicit in our proofs.} Note that $\rho(P)$ and $\delta(P)$ may only differ for some patterns~$P$ with $\delta(P) = 0$, in which case $\rho(P) \in \{0,-1\}$. 
In contrast, if $P$ contains isolated nodes (e.g., the case of Dominating Independent Set), the number of occurrences of $P$ and $t_P(n,m)$ may differ vastly (e.g., between $\Theta(n^k)$ and $\Theta(m^{k/2})$).

Our results determine that for all patterns (possibly except the triangle $K_3$), $t_P(n,m)^{1\pm o(1)}$ is the conditionally tight time complexity of detecting a Dominating $P$-Pattern if $\omega = 2$. Put differently, for any pattern $P$ (except $K_3)$, we can easily determine the conditionally optimal running time (assuming $\omega=2$)! Specifically, we obtain the following algorithmic result:

\begin{restatable}[Upper Bound]{theorem}{thmUpperBound}\label{thm:upper-bound}
    For any pattern graph $P$ with at least $16$ vertices, there is an algorithm solving Dominating $P$-Pattern problem in time $t_P(n, m)^{1 + o(1)}$.
    Further, if $\omega = 2$, this algorithm exists for all patterns $P$ except $K_3$.
\end{restatable}
We remark that our algorithms have a running time close to $t_P(n,m)^{1\pm o(1)}$ even under current values of $\omega$ and small pattern sizes -- the small polynomial overhead depends on $\omega$ and vanishes if $\omega=2$ (except for the triangle $K_3$). 
Further, we complement our algorithmic result by a corresponding conditional lower bound of $t_P(n,m)^{1-o(1)}$ based on the $k$-Orthogonal Vectors Hypothesis (and thus the Strong Exponential Time Hypothesis).

\begin{theorem}[Conditional Lower Bound]\label{thm:lower-bound}
    For any pattern graph $P$ with at least $2$ vertices, there is no algorithm solving the $P$-Dominating Set problem in time $\mathcal{O}(t_P(n, m)^{1 - \varepsilon})$ for any $\varepsilon > 0$ unless the $k$-Orthogonal Vectors Hypothesis fails.
\end{theorem}

The only pattern $P$ for which we do not obtain an algorithm matching the lower bound of Theorem~\ref{thm:lower-bound} is the triangle $K_3$. For this pattern, our best algorithm only leaves a time overhead of $m^{1/3+o(1)}$ if $\omega=2$, yielding a bound of $(t_{K_3}(n, m) \cdot m^\frac{1}{3})^{1 + o(1)} = (m^\frac{7}{3} / n)^{1 + o(1)}$ in this case. 
\begin{theorem}\label{thm:small-patterns}
 We can solve Dominating Triangle in time $\left(m^{1+\frac{2\omega}{\omega+1}}/n\right)^{1+o(1)}$. 
\end{theorem}

\paragraph*{Technical Outline and Overview}

In Section~\ref{sec:algorithm}, we obtain our algorithmic results: We first study sufficiently large connected patterns $P$  (different from $K_3$ and $K_4$), which we call \emph{basic patterns}. We argue that there are only $\bigO(t_P(n,m))$ sets that might form a Dominating $P$-Pattern (using similar arguments to Alon~\cite{Alon81}). Moreover, we show how to enumerate all such candidate sets in time $t_P(n,m)^{1+o(1)}$ in such a way that allows us to reduce to fast matrix multiplication to perform a dominance check as in Eisenbrand and Grandoni~\cite{EisenbrandG04}. Notably, this approach introduces a polynomial overhead for $K_3$ and $K_4$, as the created matrices are too large compared to the lower bound. We handle these cases separately: Using a careful combination of ideas from sparse Triangle Counting~\cite{18-triangle-counting} and sparse 2-Dominating Set~\cite{FischerKR24}, we are able to improve the algorithms for $K_4$ and $K_3$; if $\omega=2$, this completely eliminates the overhead for $K_4$ and reduces the overhead for $K_3$ to $m^{1/3+o(1)}$. Finally, we extend our algorithm beyond connected patterns by showing how to handle isolated nodes: We can no longer enumerate all candidate sets, as this number becomes too large. Instead, we apply a recursive approach similar to the algorithm for Dominating Independent Sets given in~\cite{KuennemannR24}. Here, particularly the case of a single isolated node requires great care.

Our conditional lower bound construction generalizes the ones in~\cite{KuennemannR24} (which in turn are based on~\cite{PatrascuW10, FischerKR24}): Specifically, we reduce from a version of the $k$-OV problem with carefully chosen set sizes. Intuitively, the parameter $\rho(P)$ spells out how to choose these sizes: (1) for every vertex $v\in S$, we have a set of $n$ vectors, (2) for every vertex $v\in N(S)$, we have a set of $m/n$ vectors,
and (3) for every vertex $v\in V(P)\setminus (S\cup N(S))$, we have a set of $\sqrt{m}$ vectors. The precise construction requires care; in particular if $S=\emptyset$, we need to use an alternative choice of a single set of $m/n$ vectors and $k-1$ sets of $\sqrt{m}$ vectors.


We additionally study variants of Dominating $P$-Pattern where we are given a set~$\mathcal{Q}$ of patterns rather than a single pattern. The corresponding problem Dominating $\mathcal{Q}$-Pattern asks to detect a dominating set $S$ that induces a subgraph isomorphic to \emph{some} $P\in \mathcal{Q}$. A notable special case is the Non-Induced Dominating $P$-Pattern problem in which the task is to detect a dominating set $S$ such that $G[S]$ contains $P$ as a subgraph. Combining the following theorem with Theorems~\ref{thm:upper-bound} and~\ref{thm:lower-bound}, our results settle the fine-grained complexity of Non-Induced Dominating $P$-Pattern for any pattern $P$ except $K_3$ (assuming $\omega=2$). 
\begin{theorem}[Dominating $\mathcal{Q}$-Pattern, Informal version]\label{thm:pattern-sets-informal}
    Let $Q$ be a finite set of patterns of identical order.
    The fine-grained complexity of Dominating $\mathcal{Q}$-Pattern is dominated by the pattern $P \in \mathcal{Q}$ with the highest time complexity.
\end{theorem}

\section{Preliminaries}\label{sec:preliminaries}
For a $n \in \mathbb{N}$, $[n]$ denotes the set $\{1, \dots, n\}$.
For a set $X$, its power set is denoted by $2^X$ and the set of all subsets of size $k$ by $\binom{X}{k}$.
For two sets $X,Y$, by $X \times Y$ we denote a set of unordered pairs, i.e. $X \times Y := \{\{x, y\} \mid x \in X, y \in Y\}$.
For a pattern graph $P$ and a host graph $G$, we use $k$ and $n$, respectively, to denote their order, i.e., number of vertices. As a shorthand for $\{u, v\}$, the notation $uv$ denotes an edge between $u$ and $v$.
The set $N(v) = \{u \mid uv \in E\}$ is the \emph{neighborhood} of $v$.
The \emph{closed neighborhood} of $v$ is $N[v] = N(v) \cup \{v\}$.
For a set of vertices $S \subseteq V$, by $N(S)$ (resp. $N[S]$), we denote the set $\bigcup_{v\in S}N(v)$ (resp. $\bigcup_{v\in S}N[v]$).
Moreover, the subgraph of $G$ \emph{induced} by $S$ is denoted by $G[S]$.
For a vertex $v \in V$, the graph $G - v$ denotes the graph $G$ with the vertex $v$ deleted.
Likewise, for a set $X \subseteq V$, we use $G - X := G[V \setminus X]$.
A set of vertices $D \subseteq V$ dominates the graph $G$ if every vertex of $v$ is either in $D$ or has a neighbor in $D$.

The matrix multiplication exponent $\omega \geq 2$ is the smallest constant such that there is an algorithm multiplying any two $n \times n$ matrices in $n^{\omega + o(1)}$ time.
For two rectangular matrices of size $n_1 = n \times n_2$ and $n_2 \times n_3$, $\MM(n_1, n_2, n_3)$ denotes the time complexity of multiplying these matrices.
Similarly, for matrix of size $n^a \times n^b$ and $n^b \times n^c$, $n^{\omega(a, b, c) + o(1)}$ denotes the time complexity of multiplying them.
Further, $\alpha \leq 1$ is the largest constant such that $\omega(1, \alpha, 1) = 2$.
The best known bounds for $\alpha$ and $\omega$ are $\omega\leq 2.3713$ \cite{15-faster-matrix-mult}, and $\alpha \geq 0.3213$ \cite{14-fast-matrix-mult}.
\subsection{Hardness Assumptions}\label{ssec:hardness}
Consider the $k$-Orthogonal Vectors problem ($k$-OV) stated as follows:
Given sets $A_1,\dots, A_{k}\subseteq \{0,1\}^d$, decide if there are vectors $a_1\in A_1,\dots, a_k\in A_k$ such that for all $t\in [d]$, it holds that $\prod_{i=1}^{k} a_i[t] = 0$.
A simple brute force approach solves the $k$-OV in time $\bigO(d\cdot \prod_{i\in [k]}|A_i|)$.
On the other hand, it has been a very popular hypothesis that this is the best possible way to solve this problem (up to subpolynomial factors).
In fact, it is a well-known fact that for sufficiently large~$d$ (e.g., $d=\log^2(|A_1| + \dots + |A_k|)$), any polynomial improvement over this running time would refute Strong Exponential Time Hypothesis (see \cite{VassilevskaW18,Williams05}) and has since been used as a hypothesis to show many conditional lower bounds (see e.g. \cite{abboud2018more,bringmann2024dynamic,BringmannFOCS15,buchin2019seth,chen2019equivalence} and~\cite{VassilevskaW18} for an overview).
For convenience, we state this hypothesis.
\begin{hypothesis}[$k$-OVH]\label{hyp:kovh}
    Let $n_1, n_2, \ldots, n_k \in \mathbb{N}$ be constant.
    There is no algorithm which solves all instances of $k$-Orthogonal Vector with $|A_i| = n_i$ for all $i \in [k]$ in $\mathcal{O}(\prod_{i = 1}^{k} n_i \cdot \text{poly}(d))$ time for any $\varepsilon > 0$.
\end{hypothesis}
We note that $k$-OVH is often used in the special case when $|A_1| = \dots = |A_k| = n$, however by using a simple split-and-list approach, one can show that these two hypotheses are in fact equivalent\footnote{In the sense that refuting one also refutes the other.} (see e.g. \cite{FischerKR24}).
\section{Algorithm for Pattern Dominating Set}\label{sec:algorithm}
In this section we develop the algorithms for Dominating $P$-Pattern problem for different choices of patterns $P$ and prove the following main theorem as stated in the Introduction. 
\thmUpperBound*
In order to prove this theorem, we proceed in three steps. We first develop algorithms for a class of patterns $P$ we call \emph{basic}: all patterns that contain no isolated vertices and are not isomorphic to a $K_3$ or $K_4$.
We then handle patterns with isolated vertices and the remaining small cliques $K_3, K_4$ separately. 
An important ingredient that we use to speed up the dominance check is an approach via fast matrix multiplication due to Eisenbrand and Grandoni \cite{EisenbrandG04}.
\begin{lemma}\label{lem:guessing-approach}
    Let $X,Y$ be the sets of vertices and $\phi:2^X\to \{0,1\}$ be a predicate such that for any $D\subseteq X$ we can check $\phi(D)$ in constant time.
    Let $\mathcal V_A, \mathcal V_B\subseteq 2^X$ be sets of subsets of $X$ such that any subset $D\subseteq X$ that satisfies $\phi$ can be written as a union of two sets $A\in \mathcal V_A$ and $B\in \mathcal{V}_B$.
    Then in time $\MM(|\mathcal V_A|, |Y|, |\mathcal{V}_B|)$ we can enumerate all subsets of $X$ that satisfy $\phi$ and dominate $Y$.
\end{lemma}
We remark that in the $k$-Dominating Set algorithm~\cite{EisenbrandG04}, the predicate $\phi(D)$ from the previous lemma is true if $|D|=k$ and false otherwise. In our case, for a fixed pattern $P$, the predicate $\phi(D)$ will be true if the subgraph of $G$ induced by $D$ is isomorphic to $P$ (i.e. $G[D]\cong P$).\footnote{Note that verifying $\phi$ takes $f(|P|) = \bigO(1)$ time, for constant size patterns.}
Recall also our definition of the parameter $\rho(P)$, which will be relevant in this section.
\defParameter*
Let $S$ be a (possibly empty) independent set of non-isolated vertices in $P$ that maximizes the value $|S|- |N(S)|$ as in Definition \ref{def:parameter}.
We define the \emph{remainder} set $R$ to be $V(P) \setminus N[S]$.
The following observation follows directly from the fact that $S$ is an independent set.
\begin{observation}\label{obs:s-t-partition}
    Let $P$ be a pattern on $k$ vertices and $S, N(S), R$ be as defined above. 
    Then $S, N(S),R$ form a partition of $V(P)$.
\end{observation}
Before proceeding with the algorithm, we state another observation, that gives us a nice way to think about the value of $t_P(n,m)$ depending on whether set $S$ is empty or not.
\begin{observation}\label{obs:tp-breakdown}
    Let $P$ be a pattern on $k$ vertices and $S, N(S), R$ be the partition of $V(P)$ as defined above. 
    \begin{itemize}
        \item If $S = \emptyset$, then $t_P(n,m) = \frac{m^{\left(k +1 \right)/2+o(1)}}{n}$.
        \item If $S\neq \emptyset$, then $t_P(n,m) = t_{P[N[S]]}(n,m)\cdot m^{|R|/2}$.
    \end{itemize}
\end{observation}
If we consider a partition into sets $S,N(S)$ and $R$ similar as above, intuitively we have some structure on the parts $S$ (independent set in $P$) and $N(S)$ (there exists a maximal matching that matches each vertex in $N(S)$ to a vertex in $S$), but we have very little structure on how the remainder set looks like. 
In the following lemma we show that we can decompose the remainder set $R$ into much simpler subgraphs (disjoint edges and cycles).
\begin{restatable}{lemma}{lemmaDecompR}\label{lem:decomp-T}
    Let $S$, $N(S)$ and $R$ be a partition of $V(P)$ as defined above and assume that the induced subgraph $P[R]$ contains no isolated vertices.
    Then the following holds.
    \begin{itemize}
        \item There exists a spanning subgraph of $P[R]$ that is isomorphic to a disjoint union of edges and odd cycles.
        \item For any vertex $x\in R$, there exists a spanning subgraph of $P[R-x]$ that is isomorphic to a disjoint union of edges and odd cycles.
    \end{itemize}
\end{restatable}
The proof follows from a structural theorem proved in \cite{08-edge-cycle-packing} together with the choice of the set $S$ in Definition \ref{def:parameter}.
\begin{proof}
    In \cite{08-edge-cycle-packing} the authors show that any graph $G$ has a spanning subgraph that is isomorphic to a disjoint union of edges and odd cycles if for every set $X \subseteq V$, the number of isolated vertices in $G - X$ is at most $|X|$.
    First, we look at the case where we do not delete a vertex $v$ from $T$.
    Assume a set of vertices $X \subset R$ exists such that $P[R] - X$ has more than $|X|$ isolated vertices.
    Let $Y$ be the set of isolated vertices in $P[R] - X$.
    Then we have $N_{P[R]}(Y) \subseteq X$ and $|Y| > |X|$, which implies $|Y| > |N_{P[R]}(Y)|$.
    The existence of $Y$ does however contradict the construction of $T$, because adding $Y$ to $S$ and $N_{P[R]}(Y)$ to $N(S)$ increases $|S| - |N(S)|$.
    Therefore, $X$ cannot exist and the condition for the theorem from \cite{08-edge-cycle-packing} is fulfilled.
    Now, we show that this property holds even after deleting any one vertex $v$ from $T$.
    Proceed with defining $Y$ and $X$ in the same way for $P[R] - v$ instead of $P[R]$.
    In this case, we have $N_{P[R]}(X) \subseteq Y \cup \{v\}$ and $|X| > |Y|$, which implies $|X| \geq |N_{P[R]}(X)|$.
    This contradicts the condition that the choice of $S$ maximizes its size.
    Again, $X$ cannot exist and the condition for the theorem from \cite{08-edge-cycle-packing} is fulfilled.
\end{proof}
We are now equipped with all the tools we need to construct the algorithm for the first family of patterns $P$ that we call \emph{basic patterns}.

\subsection{Basic Patterns}\label{ssec:general-algorithm}

In this section we construct an algorithm that solves the $P$-Dominating Set problem in the (conditionally) optimal time for most of the patterns. 
More precisely, we say a pattern graph~$P$ is \emph{basic} if: 1) it has at least three vertices; 2) it has no isolated vertices; 3) it is neither isomorphic to a $K_3$ nor a $K_4$.
We prove the following theorem for basic patterns.
\begin{theorem}
    For any basic pattern $P$ on at least 16 vertices, there exists an algorithm solving the Dominating $P$-Pattern problem in time $t_P(n,m)^{1+o(1)}$. Moreover, if $\omega=2$, this time complexity can be achieved for all basic patterns.
\end{theorem}
The high level idea is to first decompose the pattern into (not necessarily induced) odd cycles, edges and isolated vertices, and then enumerate all possible \emph{valid} choices for each of those parts efficiently, and then use Lemma \ref{lem:guessing-approach} to check if the union of guessed parts induces a subgraph isomorphic to $P$, and if it dominates $G$.
We first show that we can enumerate all subgraphs isomorphic to the remainder set $R$ efficiently.
Let $k:=|V(P)|$ and $n:= |V(G)|$. We say a vertex $v\in V(G)$ is \emph{heavy} if $\deg(v)\geq \frac{n}{k}-1$.
We now prove the following lemma.
\begin{restatable}{lemma}{lemmaEnumeratingT}\label{lemma:enumerating-T}
    Let $P$ be a basic pattern. Let $S,N(S),R$ be a partition of $V(P)$ as defined in Definition \ref{def:parameter}. Then the following holds.
    \begin{enumerate}
        \item We can enumerate all subgraphs of $G$ that are isomorphic to $P[R]$ in time $\bigO(m^{|R|/2})$.
        \item We can enumerate all subgraphs of $G$ that are isomorphic to $P[R]$ and contain a heavy vertex in time $\bigO(\frac{m^{(|R|+1)/2}}{n})$.
    \end{enumerate}
\end{restatable}
\begin{proof}
    Since $P$ is a basic pattern (i.e. contains no isolated vertices), by Lemma \ref{lem:decomp-T}, there exists a spanning subgraph of $R$ that is isomorphic to a disjoint union of independent edges and cycles. 
    Hence, it is sufficient to argue that we can enumerate any disjoint union of independent edges and cycles that consist of $|R|$ vertices in the claimed time, since any additional time overhead required to check if the enumerated subgraph is isomorphic to $R$ is bounded by some function $f(|R|)$ independent of the host graph $G$.
    It is clear that we can enumerate any set of $r$ independent edges in time $O(m^r)$, so it remains to prove that we can also efficiently enumerate the odd cycles.
    Let $C_{2r+1}$ be an odd cycle of size $2r+1$. We first enumerate only those subgraphs of $G$ that are isomorphic to $C_{2r+1}$ and contain a vertex of degree at least $\sqrt{m}$. 
    To this end, we notice that there are at most $\bigO(\sqrt{m})$ choices for this vertex, and moreover, the graph $C_{2r+1}-v$ has a perfect matching (for any $v$), that we can exhaustively enumerate in time $\bigO(m^{r})$, giving a time of $\bigO(m^{r+\frac{1}{2}})$ for this case.
    On the other hand, we can similarly enumerate all subgraphs of $G$ that are isomorphic to $C_{2r+1}$ and contain no vertices of degree at least $\sqrt{m}$. 
    In particular, we guess $r$ many (independent) edges with both endpoints of degree at most $\sqrt{m}$, and then the remaining missing vertex will be a neighbor of one of the guessed vertices, so we can guess it in $\bigO(\sqrt{m})$ time by enumerating all the neighbors of the guessed vertices. 
    This clearly yields a running time of $\bigO\left(f(r)m^{r+\frac{1}{2}}\right) = \bigO\left(m^{r+\frac{1}{2}}\right)$, concluding the proof of the first statement.

    For the second statement, if we want to enumerate those subgraphs of $G$ that are isomorphic to $P[R]$ \emph{and} contain a heavy vertex, we guess the heavy vertex in time $\bigO(\frac{m}{n})$ and then use that, by Lemma \ref{lem:decomp-T}, for any vertex $v$, the graph $P[R-v]$ has a spanning subgraph isomorphic to a disjoint union of independent edges and odd cycles. Then the same argument as above applies.
\end{proof}

We now use the lemma above to handle the case when the decomposition into $S$, $N(S)$ and $R$ as in Definition \ref{def:parameter} yields $S = \emptyset$. 
\begin{restatable}{lemma}{lemmaBasicSEmpty}
    Let $P$ be a basic pattern and $S, N(S), R$ be a partition of $V(P)$ as defined in Definition \ref{def:parameter}, such that $S=\emptyset$.
    If the basic pattern $P$ has at least $16$ vertices,
    then there exists an algorithm solving Dominating $P$-Pattern in time $\left(\frac{m^\frac{k - 1}{2}}{n}\right)^{1 + o(1)} = t_P(n, m)^{1 + o(1)}$.
    If $\omega = 2$, this time complexity can be achieved for all basic patterns $P$.
\end{restatable}
\begin{proof}
    If $S = \emptyset$, then $V(P) = R$ and Lemma \ref{lemma:enumerating-T} applies to all of $P$.
    We show that we can construct two sets $\mathcal V_1, \mathcal V_2$ consisting of subsets of $V(G)$, such that the size of $\mathcal V_1$ and $\mathcal V_2$ is similar, and for any subgraph $H$ of $G$ that is isomorphic to $P$ there are sets $X_1\in \mathcal V_1$ and $X_2\in \mathcal V_2$ such that $V(H) = X_1\cup X_2$.
    As observed in~\cite{FischerKR24}, any dominating set of size $k$ contains a heavy vertex $v_h$ (a vertex of degree at least $n/k-1$), hence we can first guess this vertex in time $\bigO(\frac{m}{n})$.
    By Lemma \ref{lem:decomp-T}, there is a spanning subgraph of $P$ that is isomorphic to a disjoint union of edges and odd cycles and we can find such a spanning subgraph in $f(k) = \bigO(1)$ time.
    Assume that this decomposition yields $\alpha$ odd cycles and $\beta$ disjoint edges.
    
    We first show that we can find any solution that maps all vertices that are contained in an odd cycle to a vertex in $G$ of degree at most $\sqrt{m}$. 
    We can decompose each cycle of length $2r+1$ further into $r-1$ independent edges and one subgraph isomorphic to a $P_3$ (similarly as in Lemma \ref{lemma:enumerating-T}), such that there are $\bigO(m)$ valid choices for each edge and $\bigO(m\sqrt{m})$ choices for the $P_3$.
    We can now decompose those subgraphs into two graphs $B_1$ and $B_2$ such that $B_1$ corresponds to $\lfloor\frac{\alpha}{2}\rfloor$ copies of $P_3$ and $\lceil\frac{k-1-3\alpha-2\beta}{4}\rceil$ copies of $K_2$, and
    $B_2$ corresponds to $\lceil\frac{\alpha}{2}\rceil$ copies of $P_3$ and $\lfloor\frac{k-1-3\alpha-2\beta}{4}\rfloor$ copies of $K_2$.\footnote{Note that by our construction the number $k-1-3\alpha-2\beta$ is even.}
    It is easy to see that $\left||V(B_1)|- |V(B_2)|\right|\leq 3$. 
    
    Now, if $\beta>0$, we can greedily distribute the remaining $\beta$ independent edges between $B_1$ and $B_2$ to make sure that in the end we obtain new graphs $B'_1$ and $B'_2$ that satisfy the following conditions: 
    1) without loss of generality $|V(B'_1)| \geq |V(B'_2)|$, 2) $|V(B'_1)| - |V(B'_2)|\leq 2$, 3) $B'_1\cup B'_2$ is a spanning subgraph of $P$ and 4) we can enumerate all subgraphs of $G$ that contain $B'_i$ as a spanning subgraph in time $\bigO\left(m^{\frac{|V(B'_i)|}{2}}\right)$.
    Hence, by Lemma \ref{lem:guessing-approach}, we can solve the problem in this case in time $\bigO\left(\frac{m}{n}\MM\left(m^{\frac{|V(B'_1)|}{2}}, n, m^{\frac{|V(B'_2)|}{2}} \right)\right)$.
    Now we use that $|V(B'_1)| + |V(B_2')| = k-1$ and if $\omega=2$, this clearly yields the time $\left(\frac{m^{(k-1)/2}}{n}\right)^{1+o(1)}$ for all basic patterns.
    We now show that this running time still evaluates to the one above if $k\geq 16$. 
    Note that if $k\geq 16$, we have $|V(B'_1)| + |V(B_2')| \geq 15$ and moreover $|V(B_1')| \geq |V(B_2')| \geq |V(B_1')|-2$ and from this we can conclude that $|V(B_1')|\geq |V(B_2')|\geq 7$. 
    Using that the rectangular matrix multiplication constant\footnote{For more details see Preliminaries.} $\alpha\geq 0.321$, and that $m\geq n$, we get that the exponent 
    \begin{equation}\label{eq:RMM-primes}
        \omega(1, \frac{2}{|V(B'_1)|}, 1) = \omega(1, \frac{2}{|V(B'_2)|}, 1) = 2
    \end{equation} and hence we obtain for all $k\geq 16$:
    \begin{align*}
        \bigO\left(\frac{m}{n}\MM\left(m^{\frac{|V(B'_1)|}{2}}, n, m^{\frac{|V(B'_2)|}{2}} \right)\right) & \leq  \bigO\left(\frac{m}{n}\MM\left(m^{\frac{|V(B'_1)|}{2}}, m, m^{\frac{|V(B'_2)|}{2}} \right)\right) & (m\geq n)\\
        &\leq \bigO\left(\frac{m}{n}\cdot m^{\frac{|V(B'_1)| + |V(B'_2)|}{2} + o(1)} \right) 
        & \text{(Eq. \ref{eq:RMM-primes})}\\
        & = \bigO\left(\frac{m}{n}\cdot m^{\frac{k-1}{2} + o(1)} \right) = \left(\frac{m^{\frac{k+1}{2}}}{n} \right)^{1+o(1)}.
    \end{align*}
    We now return to the case if $\beta = 0$. 
    Note that in this case the graphs $B_1$ and $B_2$ satisfy the following conditions: 
    1) (without loss of generality) $|V(B_1)|- |V(B_2)|\leq 3$, 2) $B_1\cup B_2$ is a spanning subgraph of $P$ and 3) we can enumerate all subgraphs of $G$ that contain $B_i$ as a spanning subgraph in time $\bigO\left(m^{\frac{|V(B_i)|}{2}}\right)$.
    We now construct the set $Y$ consisting of all the vertices in $G$ that are non-adjacent to $v_h$.
    Recall that for now we are only focused on finding solutions that map all vertices that are contained in an odd cycle to a vertex in $G$ of degree at most $\sqrt{m}$. 
    Hence, if $|Y|\geq (k-1)\sqrt{m} + 1$, we can conclude that no such solution exists and proceed with the next choice of heavy vertex $v_h$. 
    If on the other hand $|Y|\leq (k-1)\sqrt{m}$, by Lemma \ref{lem:guessing-approach}, we can find any such solution in time bounded by 
    \[
    \bigO\left(\frac{m}{n}\MM\left(m^{\frac{|V(B_1)|}{2}}, |Y|, m^{\frac{|V(B_2)|}{2}}\right)\right) \leq \bigO\left(\frac{m}{n}\MM\left(m^{\frac{|V(B_1)|}{2}}, \sqrt{m}, m^{\frac{|V(B_2)|}{2}}\right)\right).
    \]
    Now if $\omega=2$, by using $|V(B_1)|+|V(B_2)| = k-1$, we directly get that this time complexity is bounded by $\left(\frac{m^{(k+1)/2}}{n}\right)^{1+o(1)}$.
    On the other hand, with the current value of $\omega$, we get similarly as above for $k\geq 16$ that $|V(B_1)|\geq |V(B_1)|\geq 6$ and again since $\alpha\geq 0.321$, we obtain 
    \begin{equation}\label{eq:RMM-b1-b2}
    \omega(1,\frac{1}{|V(B_1)|},1) = \omega(1,\frac{1}{|V(B_2)|},1) = 2.
    \end{equation}
    Hence, we can bound the time complexity as
    \begin{align*}
        \bigO\left(\frac{m}{n}\MM\left(m^{\frac{|V(B_1)|}{2}}, \sqrt{m}, m^{\frac{|V(B_2)|}{2}}\right)\right) &\leq \bigO\left(\frac{m}{n}\MM\left(\left(\sqrt{m}\right)^{|V(B_1)|}, \sqrt{m}, \left(\sqrt m\right)^{{|V(B_2)|}}\right)\right) 
        \\&\leq \bigO\left(\frac{m}{n} \cdot \left(\sqrt{m}\right)^{|V(B_1)|+|V(B_2)|+o(1)} \right) & \text{(Eq \ref{eq:RMM-b1-b2})}\\
        & = \bigO\left(\frac{m}{n} \cdot \left(\sqrt{m}\right)^{k-1+o(1)} \right)  = \left(\frac{m^{\frac{k+1}{2}}}{n}\right)^{1+o(1)}        .
    \end{align*}
    It remains to show that we can also detect solutions that map a vertex contained in some odd cycle to a vertex in $G$ of degree at least $\sqrt{m}$.
    We remark that the case when $\beta>0$ can be handled very similarly as above and we omit this case and only focus on the case when $\beta=0$. 
    Here we let $0\leq \alpha'\leq \alpha-1$ be the number of odd cycles whose each vertex is mapped to a vertex in $G$ of degree at most $\sqrt{m}$. 
    We then construct the sets $B_1$ and $B_2$ similarly as before, but only on the $\alpha'$ "low-degree" odd cycles. 
    We then decompose the remaining $\alpha-\alpha'$ odd cycles into a single vertex (corresponding to the vertex that will be mapped to a "high-degree" vertex in $G$), and independent edges naturally.
    It is now easy to see that we are again in the same case as $\beta>0$, since we can distribute the remaining edges and vertices to balance the sets $V(B_1)$ and $V(B_2)$ similarly as above.
\end{proof}
The second case for basic patterns is that the decomposition into $S$, $N(S)$ and $R$ as in Definition \ref{def:parameter} yields $S \neq \emptyset$. 
Intuitively this case is simpler than the one above, since we can use the sets $S,N(S)$ to simulate $\beta>0$, even if $\beta=0$ and "balance" the sets $V(B_1)$ and $V(B_2)$. 
We only sketch the proof, as the details are to a large extent similar to the ones above.
\begin{lemma}
    Let $P$ be a basic pattern and $S, N(S), R$ be a partition of $V(P)$ as defined in Definition \ref{def:parameter}, such that $S \neq \emptyset$.
    If the basic pattern $P$ has at least $16$ vertices,
    then there exists an algorithm solving Dominating $P$-Pattern in time $\left({n^{|S|-|N(S)|}}{m^{|N(S)|+\frac{|R|}{2}}}\right)^{1 + o(1)} = t_P(n, m)^{1 + o(1)}$.
    If $\omega = 2$, this time complexity can be achieved for all basic patterns $P$.
\end{lemma}
\begin{proof}[Proof (sketch)]
    Similarly as above, construct the sets $B_1$ and $B_2$ by finding a spanning subgraph of $P[R]$ that consists of union of independent $P_3$'s and $K_2$'s and distributing the connected components as equally as possible. 
    Clearly, we have that $|V(B_1)| + |V(B_2)| = |R|$, and (without loss of generality) $0\leq |V(B_1)| - |V(B_2)| \leq 3$.
    Since $S\neq \emptyset$, also $N(S)\neq \emptyset$ (by definition $S$ contains no isolated vertices), hence we can distribute the edges from a perfect matching between the set $N(S)$ and some subset of $S$ of size $|N(S)|$ (such perfect matching always exists by Hall's marriage theorem) between the sets $B_1$ and $B_2$ similarly as in previous lemma to obtain graphs $B_1', B_2'$.
    Clearly, it holds that $|V(B'_1)| + |V(B'_2)| = 2|N(S)| + |R|$ and (without loss of generality) $0\leq |V(B'_1)| - |V(B'_2)| \leq 2$.
    We then define the graphs $Q_1, Q_2$ that consist of $\lfloor\frac{|S|-|N(S)|}{2}\rfloor$ and $\lceil\frac{|S|-|N(S)|}{2}\rceil$ isolated vertices respectively.
    Similarly as above, we can enumerate all subgraphs of $G$ isomorphic to $B'_i\cup Q_i$ in time $\bigO\left(n^{|Q_i|}m^{\frac{|B'_i|}{2}}\right)$.
    Also, by similar arguments as above, we have that for $k\geq 16$, the matrices are "thin enough" that rectangular matrix multiplication techniques can be used to obtain the tight running time of 
    \begin{align*}
        \MM\left(n^{|Q_1|}m^{\frac{|B'_1|}{2}}, n, n^{|Q_2|}m^{\frac{|B'_2|}{2}} \right) &\leq \left(n^{|Q_1|+|Q_2|}m^{(|B'_1|+|B'_2|)/2}\right)^{1+o(1)}&(k\geq 16)\\
        &\leq \left(n^{|S|-|N(S)|}m^{\frac{2|N(S)|+|R|}{2}}\right)^{1+o(1)} \qedhere
    \end{align*}
\end{proof}


\subsection{Small Cliques}\label{ssec:small-graphs-algorithms}
In this section, we deal with the remaining small patterns, i.e., the cliques of size at most 4. We note that in \cite{KuennemannR24} the $K_2$ was already settled, so it remains to settle the cliques of size $3$ and $4$.
In particular, this section is dedicated to proving the following two main theorems. 
\begin{theorem}\label{thm:triangle-algo}
    For a Dominating Triangle, there exists a deterministic $\left(\displaystyle\frac{m^{1+\frac{2\omega}{\omega+1}}}{n}\right)^{1+o(1)}$-time algorithm and a randomized $m^{\frac{\omega+1}{2}+o(1)}$-time algorithm.
\end{theorem}
\begin{restatable}{theorem}{thmCliqueDSAlgo}\label{thm:K4-DS}
    There is a randomized algorithm solving $K_4$-Dominating Set in time $\left(\frac{m^{(\omega+3)/2}}{n}\right)^{1+o(1)}$. If $\omega=2$, this is $\left(\frac{m^{5/2}}{n}\right)^{1+o(1)} = t_{K_4}(n,m)^{1+o(1)}$.
\end{restatable}

We start with the simple deterministic algorithm for Dominating Triangle problem that uses a reduction to the \emph{All-Edges-Triangle-Counting} problem.
This problem asks for a given tripartite graph $G = (V_1 \cup V_2 \cup V_3, E)$ to determine, for every edge $e$ in $(V_1 \times V_2) \cap E$, how many triangles in $G$ contain $e$.
It is well known that this problem can be solved in time $m^{\frac{2\omega}{\omega + 1} + o(1)}$ (see e.g. \cite{18-triangle-counting}). We remark that any $\bigO(m^{\frac{4}{3}-\varepsilon})$-time algorithm would refute the $3$-SUM and the APSP Hypothesis~\cite{Patrascu10, VassilevskaWX20}, giving matching lower bounds if $\omega=2$.
\begin{lemma}[All-Edge Triangle Counting, see~\cite{18-triangle-counting}]\label{lem:triangle-counting-time}
    There is an algorithm solving All-Edges-Triangle-Counting in time $m^{\frac{2\omega}{\omega + 1} + o(1)}$.
\end{lemma}
By using this lemma, we are able to solve the problem in $\bigO\left(t_{K_3}(n,m)\right)\cdot m ^{\frac{\omega-1}{\omega+1}+o(1)}$, which evaluates to $\bigO\left(t_{K_3}(n,m)\right)\cdot m ^{\frac{1}{3}+o(1)}$ if $\omega=2$. 
By a more careful approach, we are able to construct a randomized algorithm that uses a bloom-filter inspired approach similar to the one in \cite{FischerKR24}, to achieve $\bigO\left(t_{K_3}(n,m)\right)\cdot \frac{n}{\sqrt{m}}$, which is better whenever $m\in \Omega(n^{1.10283})$ with current value of $\omega$ and $m\in \Omega(n^{1.2})$ if $\omega=2$, giving us the second part of the Theorem \ref{thm:triangle-algo}. 
We remark that although this value can be slightly improved by using a more careful analysis, it seems that in order to match the lower bound, a new technique will be required.
We now provide the approach of the deterministic algorithm and remark that the details for the randomized algorithm will follow from the proof of Theorem \ref{thm:K4-DS}. 
\begin{restatable}{proposition}{propDetDomTriangle}\label{prop:dom-triangle}
    There exists a deterministic algorithm solving Dominating Triangle in time $\left(\displaystyle\frac{m^{1+\frac{2\omega}{\omega+1}}}{n}\right)^{1+o(1)}$. If $\omega=2$, this equals $\left(\frac{m^{7/3}}{n}\right)^{1+o(1)}$.
\end{restatable}
\begin{proof}
    Guess a heavy vertex $v_h\in \mathcal{H}$. Let $X_0,X_1,X_2$ be the first three layers in a BFS layering from $v_h$. If $V\neq X_0\cup X_1\cup X_2$, we conclude that there is no dominating $K_3$ that contains $v_h$ and we proceed with the next heavy vertex.
    Otherwise, any dominating $K_3$ that contains $v_h$ faithfully corresponds to an edge $\{u,v\}\in X_1$ such that $X_2\subseteq N(u)\cup N(v)$. 
    More preciselly, let $N_2(u):=N(u)\cap X_2$. Then any dominating $K_3$ that contains $v_h$ faithfully corresponds to an edge $\{u,v\}\in X_1$ such that $ |N_2(u)\cup N_2(v)| = |X_2|$.
    Using the principle of inclusion-exclusion, we can write $|N_2(u)\cup N_2(v)| = |N_2(u)| + |N_2(v)| + |N_2(u)\cap N_2(v)|$. We can in linear time precompute the value $|N_2(x)|$ for each $x\in X_1$, but it remains to compute the value of $|N_2(u)\cap N_2(v)|$.
    By noticing that this value is only relevant when $u,v\in X_1$ form an edge, we can observe that this value corresponds exactly to the number of triangles of type $u,v,x$ where $x\in X_2$. 
    To this end we employ the reduction to All-Edges-Triangle-Counting to conclude the algorithm.
    Let $G_h$ be the tripartite graph whose vertex set consists of two copies of $X_1$ and one copy of $X_2$, i.e. $V = X_1^{(1)}\cup X_1^{(2)}\cup X_2$ and the edges are added naturally, that is, there is an edge between two vertices $u,v$ in different parts if and only if this edge is present in $G$.
    It is easy to see that for any $u\in X_1^{(1)}$ and $v\in X_1^{(2)}$, such that $\{u,v\}\in E(G_h)$ it holds that the number of triangles in $G_h$ that contain the edge $u,v$ is equal to $|N_2(u)\cap N_2(v)|$. 
    Hence, by running the algorithm from Lemma \ref{lem:triangle-counting-time}, we can compute this value for all edges in time $m^{2\omega/(\omega+1)+o(1)}$.

    It remains to argue the time complexity of this algorithm. There are $\bigO\left(\frac{m}{n}\right)$ many choices for a heavy vertex $v_h$ in $G$. For each of them, we run BFS and construct a graph $G_h$ ($\bigO(m)$), and finally run the All-Edge-Triangle-Counting algorithm from Lemma \ref{lem:triangle-counting-time} ($m^{2\omega/(\omega+1)+o(1)}$). In total, this clearly gives us the desired time complexity.
\end{proof}

On the other hand, by coming up with a more involved argument, we are able to construct a randomized algorithm that matches the lower bound (up to resolving the matrix multiplication exponent $\omega$) for the remaining small pattern, namely Dominating $4$-Clique. 
The algorithm starts similarly as above.
That is, we first guess a heavy vertex $w$ in time $\bigO\left(\frac{m}{n}\right)$, and then we partition the vertices into two sets $X_1$ and $X_2$ based on the distance from $w$ and thus reduce to
detecting a triangle in $X_1$ that dominates $X_2$. 
We then consider two cases based on the size of $X_2$.
If $|X_2|\leq 3\sqrt{m}$, this turns out to be quite easy, and boils down to using standard matrix multiplication algorithms (with a bit more careful analysis).
If on the other hand $|X_2|> 3\sqrt{m}$, we then either reduce to counting cliques (triangles and $K_4$'s) in a certain restricted setting, or use the randomized bloom-filter based approach of \cite{FischerKR24} (with a finer-grained analysis) to achieve the desired running time.
We dedicate the rest of this section to giving the full proof of Theorem \ref{thm:K4-DS}. 
\thmCliqueDSAlgo*
As earlier, we first guess the heavy vertex in time $\bigO\left(\frac{m}{n}\right)$ and then for each choice of heavy vertex $v$, we partition the graph into two sets $X_1$ and $X_2$ consisting of vertices at distance one and two from $v$ respectively. Recall that if there is a vertex $w\in V\setminus (X_1\cup X_2)$, we can immediately conclude that no dominating clique contains $v$ and continue with the next choice of the heavy vertex. 
We now reduce this problem to finding a triangle in $X_1$ that dominates $X_2$. 
In the rest of this section we discuss how to find such a triangle efficiently.
But first, we introduce some useful notation. Let $V_\ell := \{x\in X_1\mid \deg(x)\leq \sqrt{m}\}$, and $V_h := \{x\in V_1\mid \deg(x)>\sqrt{m}\}$. We will refer to the vertices in $V_\ell$ (resp. $V_h$) as \emph{low-degree} (resp. \emph{high-degree}) vertices.
We now consider two cases based on the size of the set $X_2$. 
\paragraph*{Case 1: $|X_2|\leq 3 \sqrt m$.}
We consider two simple sub-cases based on whether there is a solution that contains a high-degree vertex. 
We start with the simple case for checking solutions that contain a high degree vertices. 
\begin{claim}\label{claim:dom-k4-small-x2}
    For $|X_2|\in \bigO(\sqrt{m})$, there is an algorithm that checks if there exists a triangle $u,v,w$ in $X_1$ that dominates $X_2$, with $u\in V_h$ in time bounded by $\MM\left(\sqrt{m}, m, \sqrt m\right)$.\footnote{With state-of-the-art matrix multiplication, this is bounded by $m^{1.6252+o(1)}$ (see \cite{14-fast-matrix-mult})} If $\omega=2$, this time is bounded by $m^{3/2+o(1)}$.
\end{claim}
\begin{proof}
    We simply consider all possible choices for $u$ ($|V_h|$ many choices) and since $u,v,w$ form a triangle, we consider all the possible choices for $v,w$ forming an edge. Using Lemma \ref{lem:guessing-approach}, we directly get the running time $\MM\left(|V_h|, |X_2|, m\right)$. Plugging in $|X_2|\leq \sqrt{m}$ and $|V_h| \leq \sqrt{m}$, we obtain the desired running time.
    It only remains to prove that if $\omega=2$, this time is bounded by $m^{3/2+o(1)}$. 
    This follows directly by using $\MM(a,b,c) \overset{\omega = 2}{=} \max(ab,ac,bc)$.
\end{proof}
We now proceed to construct an algorithm for the remaining case, namely finding solutions that contain no high-degree vertices.
\begin{claim}
    There is an algorithm that checks if there exists a triangle $u,v,w$ in $X_1$ that dominates $X_2$ with $u,v,w\in V_\ell$ in time 
    $m^{(\omega+1)/2+o(1)}$. If $\omega=2$, this is bounded by $m^{3/2+o(1)}$. 
\end{claim}
\begin{proof}
    Clearly, in any solution, there exists a vertex $u$ such that $\deg(u)\geq |X_2|/3$. 
    Moreover, since we are only looking for the solutions where each solution vertex is contained in $V_\ell$, we have $|X_2|/3\leq \deg(u)\leq \sqrt{m}$.
    We now proceed by guessing this vertex and then using the matrix multiplication approach as in Lemma \ref{lem:guessing-approach} to obtain the following running time.
    \begin{align*}
        T(m,n) &\leq \sum_{u\in V_\ell}\MM\left(\deg(u), |X_2|, \deg(u)\right) \\& \leq \sum_{u\in V_\ell}\MM\left(\deg(u), 3\deg(u), |N(u)|\right) & (|X_2|\leq 3\deg(u)) \\ &\leq \sum_{u\in V_\ell}\deg(u)^{\omega+o(1)}\\
        & \leq m^{(\omega-1)/2+o(1)}\sum_{u\in V_\ell}\deg(u) & (\deg(u)\leq \sqrt m) \\
        & \leq m^{(\omega-1)/2+o(1)}\cdot 2m \\
        &= m^{(\omega + 1)/2+o(1)}. & &\qedhere
    \end{align*}
\end{proof}
By combining the last two claims, we get the desired time.
\begin{corollary}\label{cor:dom-k4-small-x2}
    If $|X_2|\leq 3\sqrt{m}$, we can check if there exists a triangle $u,v,w\in X_1$ that dominates $X_2$ in time $m^{(\omega+1)/2+o(1)}$.
\end{corollary}
\begin{proof}
    It is sufficient to bound the time from Claim \ref{claim:dom-k4-small-x2} by $m^{(\omega+1)/2+o(1)}$.
    \[
    \MM(\sqrt{m},m\sqrt{m})\leq \bigO\left(\sqrt{m}\right)\MM(\sqrt{m},\sqrt{m},\sqrt{m})\leq \bigO\left(\sqrt{m}\right)m^{\omega/2+o(1)} \leq m^{(\omega+1)/2+o(1)}.
    \]
\end{proof}
\paragraph*{Case 2: $|X_2|> 3 \sqrt m$.}
We now move on to a technically more interesting case. For each $0\leq i \leq 3$, we will separately construct algorithms that detect a triangle $u,v,w$ in $X_1$ that dominates $X_2$ and such that $|\{u,v,w\}\cap V_\ell| = i$. 
We start with the easiest case, when $|\{u,v,w\}\cap V_\ell| = 3$. 
Since all three vertices belong to $V_\ell$, their respective degrees are at most $\sqrt{m}$, hence they can dominate at most $3\sqrt{m}$ vertices in total. However, since we are in the case $|X_2|>3\sqrt{m}$, it follows that no such solution can exist. We formally state this as an observation.
\begin{observation}
    For $|X_2|>3\sqrt{m}$, there is no triangle $u,v,w$ that dominates $|X_2|$, such that $|\{u,v,w\}\cap V_\ell| \geq 3$.
\end{observation}
On the other end of the spectrum is another relatively easy case, namely $\{u,v,w\}\cap V_\ell = \emptyset$. 
Here we proceed similarly as in the proof of Proposition \ref{prop:dom-triangle}, by reducing to All-Edges-Triangle-Counting.
\begin{lemma}
    There is an algorithm that checks if there exists a triangle $u,v,w$ in $X_1$ that dominates $X_2$ with $u,v,w\in V_h$ in time 
    $m^{\frac{\omega+1}{2}+o(1)}$. If $\omega=2$, this is bounded by $m^{3/2+o(1)}$.
\end{lemma}
\begin{proof}
    We first observe that in order for the triangle $u,v,w$ to dominate $X_2$, (without loss of generality) $\deg(w)\geq |X_2|/3$, and hence there are only $\bigO(\frac{m}{|X_2|})$ choices for the vertex $w$. 
    Also, since $|V_h|\leq \bigO(\sqrt{m})$, there are only at most $\bigO(\frac{m^2}{|X_2|})$ many possible choices for a suitable triple $u,v,w$.
    Recall that for any $x\in X_1$, by $N_2(x)$ we denote the set $N(x)\cap X_2$. By the principle of inclusion-exclusion, any triple $u,v,w$ dominates precisely 
    \begin{align}\label{eq:Duvw}
        D(u,v,w) &= |N_2(u)|+|N_2(v)|+ |N_2(w)|\\& - \left(|N_2(u)\cap N_2(v)| + |N_2(u)\cap N_2(w)|+ |N_2(v)\cap N_2(w)|\right) \\& + |N_2(u)\cap N_2(v)\cap N_2(w)|
    \end{align}
    many vertices in $X_2$.
    Similarly as in Proposition \ref{prop:dom-triangle}, we notice that the number $|N_2(u)\cap N_2(v)|$ is precisely equal to the number of triangles $u,v,y$ for $y\in X_2$
    (symmetrically for the pairs $u,w$ and $v,w$), and hence by Lemma \ref{lem:triangle-counting-time} we can compute it for all pairs in time $m^{\frac{2\omega}{\omega+1}+o(1)}$.
    It remains to compute the number $|N_2(u)\cap N_2(v)\cap N_2(w)|$ for all triples. We note that this number is equal to the number of cliques of size $4$ of type $u,v,w,y$ with $y\in X_2$.
    By using the classical rectangular matrix multiplication algorithm, we can compute, for all suitable triples $u,v,w$, the number of such cliques in time $\MM(\frac{m\sqrt{m}}{n}, N, \sqrt{m})$.
    We now proceed to bound this time by the desired value.
    \begin{align*}
        \MM(\frac{m\sqrt{m}}{N}, N, \sqrt{m}) &\leq \bigO\left(\frac N {\sqrt m}\right) \MM(\frac{m\sqrt m}{N}, \sqrt m, \sqrt{m}) & (N\ge \sqrt{m})\\
        & \leq \bigO\left(\sqrt m\right) \MM(\sqrt{m},\sqrt{m}, \sqrt{m}) & (m\ge N) \\
        & = \bigO\left(\sqrt m\right) m^{\omega/2 + o(1)}\\
        &\leq m^{\frac{\omega+1}{2}+o(1)}
    \end{align*}
    Finally, we get the following full algorithm. 
    First compute for all edges $u,v\in V_h$ the number of triangles that contain $u,v$ and an endpoint in $X_2$ in time $m^{\frac{2\omega}{\omega+1}+o(1)}$, and for all triangles $u,v,w$ with $u,v\in V_h$, $\deg(w)\geq \frac{N}{3}$ the number of $K_4$'s that contain an endpoint in $X_2$ in time $m^{\frac{\omega+1}{2}+o(1)}$. 
    Then iterate over each such triangle ($\bigO(\frac{m^{5/2}}{N})\leq \bigO(m^{3/2})$ many), and for each such triangle $u,v,w$ in constant time fetch the value $D(u,v,w)$ as in Equation \ref{eq:Duvw} and check if this value is $\geq |X_2|$. 
    In total, this running time is clearly dominated by the term $m^{\frac{\omega+1}{2}+o(1)}$.
\end{proof}
Before proceeding with the remaining two cases, we need to first introduce an intermediate problem, as studied in the context of finding dominating sets of size two in sparse graphs efficiently \cite{FischerKR24}. 
\begin{definition}[Max-Entry Matrix Product]
    Consider $0-1$ matrices $B$ of size $p\times N$ and $C$ of size $N\times q$, where $C$ contains at most $pq$ nonzeros. The \emph{Max-Entry Matrix Product} problem is to decide whether there exists $i,j$ such that
    \[
    \left(B\cdot C\right)[i,j] = \sum_{k}B[i,k]\cdot C[k,j] = \sum_k B[i,k].
    \]
\end{definition}
As already discussed in \cite{FischerKR24}, this problem is a special case of the \emph{Subset Query} problem, which has been studied in previous work (e.g. \cite{GaoIKW17,CharikarIP02}). 
In \cite{FischerKR24}, they also give a randomized algorithm that uses a Bloom-filter inspired construction to solve a special case when $p = \frac{m}{N}$ and $q=N$ in time $m^{\omega/2 + o(1)}$. 
For convenience, we will formally state this result as a lemma, but will also use a finer analysis in some other cases.
\begin{lemma}\label{lemma:max-entry-mm}
    Let $B$ be an $\bigO(\frac{m}{N})\times N$ binary matrix and $C$ be an $N\times N$ binary matrix with at most $m$ non-zeros. Then there exists a randomized algorithm computing the Max-Entry Matrix Product of $B$ and $C$ in time $m^{\omega/2+o(1)}$. Moreover, this algorithm correctly enumerates all witnesses w.h.p.
\end{lemma}
Furthermore, another useful corollary of the Lemma 3.4 in \cite{FischerKR24} is the following.
\begin{lemma}\label{lemma:gen-max-entry-mm}
    Let $B$ be a $p\times N$ binary matrix and $C$ be an $N\times q$ binary matrix with at most $pq$ non-zeros, such that each column of $C$ has at most $\Delta$ non-zeros. Then there exists a randomized algorithm computing the Max-Entry Matrix Product of $B$ and $C$ in time
    \[
        T(p,q,\Delta) \leq  \tilde\bigO\left(\max_{1\leq d\leq \Delta}\MM\left( p,d,\min\left\{ q, \frac{pq}{d}\right\}\right)  \right).
    \]
    Moreover, this algorithm correctly enumerates all witnesses w.h.p.
\end{lemma}
Our strategy for solving the remaining two cases is to construct an algorithmic reduction to the Max-Entry-Matrix-Product (using a similar approach as in \cite{FischerKR24}), and then solve this problem efficiently. 
\begin{lemma}
    There is a randomized algorithm that checks if there exists a triangle $u,v,w$ in $X_1$ that dominates $X_2$ with $u\in V_\ell$, $v,w\in V_h$ in time $m^{\frac{\omega+1}{2}+o(1)}$.
\end{lemma}
\begin{proof}
    Let $N:=|X_2|$.
    We consider two subcases based on the degree of $u$. 
    \begin{claim}
        There is a randomized algorithm that checks if there exists a triangle $u,v,w$ in $X_1$ that dominates $X_2$ with $\frac{2m}{N}\leq \deg(u)\leq \sqrt{m}$, $v,w\in V_h$ in time $m^{\frac{\omega+1}{2}+o(1)}$.
    \end{claim}
    \begin{proof}[Proof (of the claim)]
        Since the degree of $u$ is at least $\frac{2m}{N}$, there are at most $N$ choices for $u$.
        Let $C$ be an $N\times N$ binary matrix where columns correspond to the possible choices of $u$ (if there are fewer than $N$ such choices, the remaining columns are understood to be all-zero vectors) and rows correspond to the vertices in $X_2$.
        Set the value of $C[i,j]$ to one if the corresponding vertices are adjacent and to $0$ otherwise (i.e. $C$ can be thought of as a restricted adjacency matrix).
        Clearly, $C$ has at most $m$ many non-zeros.
        We now use that for any possible triple $u,v,w$, (without loss of generality) $\deg(v)\geq \frac{N}{3}$, and hence there are at most $\bigO(\frac{m}{N})$ choices for the vertex $v$. Moreover, since $w\in V_h$, there are at most $\bigO(\sqrt{m})$ many choices for $w$.
        Now, let $B$ be a $\bigO\left(\sqrt{m}\cdot \frac{m}{N}\right)\times N$ matrix, where columns correspond to vertices in $X_2$, and each row corresponds to a possible choice for a pair of vertices $v,w$. 
        Set the value of $B[(v,w),j]$ to $0$ if $j$ is adjacent to either $w$ or $v$, and to $1$ otherwise.
        It is now easy to see that a triple $u,v,w$ forms a dominating set if and only if 
        $\left(B\cdot C\right)[(v,w),u]  = \sum_k B[(v,w),k]$ (intuitively each vertex in $X_2$ that is \emph{not} dominated by $v$ nor $w$, is dominated by $u$).
        By Lemma \ref{lemma:max-entry-mm}, it is easy to see that we can compute this value for all triples w.h.p. in time $\bigO(\sqrt{m})m^{\omega/2+o(1)} = m^{\frac{\omega+1}{2}+o(1)}$.
        Finally, it remains to go over all the witnesses (at most $m\sqrt{m}$ many) and for each check if the corresponding vertices form a triangle, which can be done in constant time.
    \renewcommand{\qedsymbol}{$\vartriangleleft$}
    \end{proof}
    It remains to verify the case when degree of $u$ is at most $\bigO(\frac{m}{N})$.
    \begin{claim}
        There is a randomized algorithm that checks if there exists a triangle $u,v,w$ in $X_1$ that dominates $X_2$ with $\deg(u)\leq \frac{2m}{N}$, $v,w\in V_h$ in time $m^{\frac{\omega+1}{2}+o(1)}$.
    \end{claim}
    \begin{proof}[Proof (of the claim)]
        We assume (without loss of generality) that $\deg(w)\geq N/3$.
        Let $t\geq \sqrt{m}$ be such that $t\leq \deg(w)\leq 2t$. 
        There are clearly at most $\bigO(\frac{m}{t})$ many choices for the vertex $w$. 
        Iterate over each such choice for vertex $w$ and label the vertices in $X_2$ that are adjacent to $w$ as dominated. 
        Let $X_2^w$ be the set of non-dominated vertices by any fixed $w$ and write $N_w$ to denote the number of non-dominated vertices in $X_2^w$.
        We now use that any solution vertex $u$ must be adjacent to $w$ (solution forms a clique), and hence there are at most $\bigO(t)$ choices for $u$. 
        We proceed construct the matrices $B$ and $C$ similarly as above.
        Let $C$ be an $N_w\times \bigO\left(t\right)$ binary matrix where columns correspond to the possible choices of $u$ and rows correspond to the vertices in $X_2^w$, and set the value of $C[i,j]$ to one if the corresponding vertices are adjacent and to $0$ otherwise.
        Notice that each column of $C$ has at most $\frac{2m}{N}$ many ones.
        Let $B$ be an $\bigO(\frac{m}{N}) \times N_w$ matrix, where each row corresponds to a choice of the vertex $v$ and the columns correspond to vertices in $X_2^w$.
        Set the value of $B[v,j]$ to $0$ if $j$ is adjacent to $v$, and to $1$ otherwise.
        As above, a triple $u,v,w$ forms a dominating set if and only if 
        $\left(B\cdot C\right)[v,u]  = \sum_k B[v,k]$. Hence it is sufficient to compute the Max-Entry Matrix Product on $(B,C)$.
        Since $N_w=\bigO(N)$, for simplicity, we will use $N$ to denote $N_w$.
        By using Lemma \ref{lemma:gen-max-entry-mm}, we can bound the time complexity for computing this product as follows.
        \begin{align*}
            T_w(m,N) &\leq \tilde\bigO\left(\max_{1\leq d\leq \frac{2m}{N}}\MM\left( \frac{m}{N},d,t\right)  \right) \\
            & \leq \MM\left( \frac{m}{N}, \frac{m}{N}, t \right) \\
            & \leq \MM\left(\sqrt{m}, \sqrt{m}, t\right) & (N\geq \sqrt{m})\\
            & \leq \bigO\left(\frac{t}{\sqrt{m}}\right)\MM\left(\sqrt{m}, \sqrt{m}, \sqrt{m}\right) & (t\geq \sqrt{m})\\
            & = \bigO\left(\frac{t}{\sqrt{m}}\right) m^{\omega/2 + o(1)}
        \end{align*}
        Now, we run this product for each choice of $w$ with degree roughly $t$. Hence for a fixed $t$, we need to compute the matrix product $\bigO(\frac{m}{t})$ many times, and hence we have
        \begin{align*}
        T(m,N,t) & \leq \sum_{\substack{w\in X_1 \\ \deg(w)\in[t,2t]}} T_w(m,n) \\
        & \leq \bigO\left(\frac{m}{t}\right)T_w(m,n)\\
        &\leq \bigO\left(\sqrt m\right) m^{\omega/2 + o(1)} \\&\leq m^{(\omega+1)/2 + o(1)}.
        \end{align*}
        Finally, since the algorithm above yields the desired time for any fixed value of $t$, in order to cover the whole search space for $w$, we run the routine above $s \leq \bigO(\log m)$ times, once for each $t_1,\dots, t_s$, with $t_1 = \sqrt{m}$, $t_s = n$, and $t_{i+1} = 2t_i$. 
        This yields a total running time of
        \begin{align*}
            T(m,N) & = \sum_{i\in [s]} T(m,N,t_i)\\
            & \leq \bigO\left(\log m\right) T(m,n,t_1)\\
            & \leq m^{(\omega+1)/2 + o(1)} & (\log m = m^{o(1)})
        \end{align*}
        Hence, the claimed running time is correct.
        \renewcommand{\qedsymbol}{$\vartriangleleft$}
    \end{proof}
    The last two claims conclude the proof of the lemma.
\end{proof}
Finally, it remains to prove that we can achieve the desired running time for $|\{u,v,w\}\cap V_\ell| = 2$.
\begin{lemma}\label{lemma:k4-2-light-vertices}
    There is a randomized algorithm that checks if there exists a triangle $u,v,w$ in $X_1$ that dominates $X_2$ with $u,v\in V_\ell$, $w\in V_h$ in time $m^{\frac{\omega+1}{2}+o(1)}$.
\end{lemma}
\begin{proof}
    We will construct matrices $B$ and $C$ similarly as above. Let $B$ be an $\bigO(\sqrt{m})\times N$ matrix where each row corresponds to a vertex in $V_h$ and each column to a vertex in $X_2$. Set the value of $B[w,x]$ to $0$ if $\{w,x\}\in E$ and to $1$ otherwise. 
    Similarly, let $C$ be an $N\times \bigO(m)$ matrix, such that each row corresponds to a vertex in $X_2$ and each column to an edge $uv$ where both $u$ and $v$ are contained in $V_\ell$.
    We set an entry $C[x,(u,v)]$ to $1$ if $x$ is adjacent to either $u$ or $v$ and to $0$ otherwise.
    It is easy to verify that $u,v,w$ dominate $X_2$ if and only if $(B\cdot C)[w, (u,v)] = \sum_k B[w,k]$.
    Hence, computing the Max-Entry Matrix Product on $(B,C)$ suffices. To this end, we notice that each column of $C$ corresponds to an edge with both endpoints in $V_\ell$ and hence has at most $2\sqrt{m}$ many non-zeros.
    Thus, by Lemma \ref{lemma:gen-max-entry-mm}, we can bound the complexity of computing Max-Entry Matrix Product on $(B,C)$ as follows.
    \begin{align*}
        T(m,N) &\leq \tilde \bigO\left( \max_{1\leq d\leq 2\sqrt{m}} \MM \left(\sqrt{m}, d, \min\left\{m, \frac{m^{3/2}}{d}\right\}\right)\right) \\
        &\leq \tilde \bigO\left( \max_{1\leq d\leq 2\sqrt{m}} \MM \left(\sqrt{m}, d, m,\right\}\right) \\
        & \leq \MM\left( \sqrt{m}, \sqrt{m}, m \right) & (d\leq \bigO(\sqrt{m}))\\
        &\leq \bigO(\sqrt{m}) m^{\omega/2+o(1)} \\
        & \leq  m^{(\omega+1)/2+o(1)}
    \end{align*}\qedhere
\end{proof}
Theorem \ref{thm:K4-DS} now follows by running the corresponding algorithm for each of the $\bigO\left(\frac{m}{n}\right)$ heavy vertices.
Moreover, a direct consequence of the algorithm for case 2, by setting $|X_2| = n$ is a randomized algorithm for Dominating Triangle that runs in $m^{(\omega+1)/2+o(1)}$ and that together with Proposition \ref{prop:dom-triangle} gives the proof of Theorem \ref{thm:triangle-algo}.

\subsection{Handling Isolated Vertices}\label{ssec:isolated-vertices-algorithm}
In the last two subsections we considered the graphs that have no isolated vertices. It remains to prove that we can obtain a tight classification even for patterns with isolated vertices.
In particular, we prove the following theorem.
\begin{theorem}
    Let $P$ be any pattern graph with $k$ vertices and $1\leq r\leq k$ isolated vertices. There is a randomized algorithm that enumerates all dominating $P$-patterns in time $t_P(n,m)\cdot m^{\frac{\omega-2}{2}+o(1)}$ with high probability. If $\omega=2$, this is $t_P(n,m)^{1+o(1)}$.
\end{theorem}
We again consider two cases separately, namely when $r\geq 2$ and $r=1$. 
In the first (simpler) case of $r \ge 2$, we aim to reduce to the Maximal $r$-Independent Set problem, which is known to have efficient algorithm in sparse graphs (see \cite{KuennemannR24}). 
We enumerate all subgraphs of $G$ that correspond to the induced subgraph of $P$ that does not contain isolated vertices, and then after removing the closed neighborhood of the enumerated subgraph, we run the Maximal $r$-Independent Set algorithm on the remaining part. 
By a careful analysis, we obtain the desired time.
\begin{restatable}{lemma}{lemmaManyIsolated}
    Let $P$ be a pattern with $2\leq r\leq k$ isolated vertices. Then there exists a randomized algorithm solving $P$-Dominating Set in time $\left(t_P(n,m)\cdot m^{(\omega-2)/2}\right)^{1+o(1)}$ with high probability. If $\omega=2$, this is equal to $(t_P(n,m))^{1+o(1)}$.
\end{restatable}
\begin{proof}
    If $r=k$, then this problem is precisely the Maximal $k$-Independent Set problem, and by \cite{KuennemannR24} can be solved in the time bounded by $\left(\frac{m^{(k-1)/2+\omega/2}}{n}\right)^{1+o(1)}$, as desired. Hence, we may assume that $2\leq r\leq k-2$ (if $P$ contains at least one edge, then at least two vertices have degree at least one).
    Let $X:=\{x\in V(P): \deg_P(x)\geq 1\}$ and $Y:=V(P)\setminus X$ be the set of non-isolated and isolated vertices in $P$ respectively.
    A crucial observation that we use is that if $D\subseteq V(G)$ is a $P$-dominating set in $G$, 
    if we let $v_x,v_y$ denote the vertices that correspond to $x$ and $y$ in $D$ respectively, then the distance between $v_x$ and $v_y$ in $G$ is at least two, and moreover any vertex that is adjacent to $v_x$ is not necessarily adjacent to $v_y$. 
    Hence, if we can efficiently guess the vertices in $D$ that correspond to $X$ (denote this set by $D_X$), we only need to solve the Maximal $|Y|$-Independent Set problem on the graph $G-N[D_X]$.

    Let $S,N(S),R$ be the partition of $V(P)$ as in Definition \ref{def:parameter}. 
    If $S\neq \emptyset$, by \cite{NgoPRR18}, we can enumerate all induced subgraphs of $G$ isomorphic to $P[X]$ in time $\bigO(t_{P[X]}(n,m))$. 
    By a very simple modification of the algorithm presented in \cite{KuennemannR24}, there is a randomized algorithm solving Maximal $|Y|$-Independent Set problem on $G-N[D_X]$ in time bounded by $m^{(|Y|-2+\omega)/2+o(1)}$.
    In total, this yields a time complexity of 
    \begin{align*}
    T(m,n) &\leq t_{P[X]}(n,m)\cdot \left(m + m^{(|Y|-2+\omega)/2+o(1)}\right)\\ &\leq \left(t_{P[X]}(n,m) \cdot m^{\frac{|Y|-2+\omega}{2}}\right)^{1+o(1)} & (|Y|\geq 2)\\
    & {=}\left(T_P(n,m) \cdot m^{(\omega-2)/2}\right)^{1+o(1)} & \text{(Observation \ref{obs:tp-breakdown})}
    \end{align*}
    On the other hand, if $S=\emptyset$, then we use that any dominating set must contain a heavy vertex in $G$. 
    Moreover, $P[X]$ is now either isomorphic to $K_3$ or $K_4$, or is a basic. If it is a basic, by Lemma \ref{lemma:enumerating-T}, we can enumerate all subgraphs of $G$ that are isomorphic to $P[X]$ and contain a heavy vertex in time $\bigO\left( \frac{m^{(|X|+1)/2}}{n}\right) = \bigO\left(t_{P[X]}(m,n)\right)$. It is straightforward to verify that in both of the remaining cases, the same holds. 
    Now by the same argument as above, we can bound the total time complexity in this case by 
    \begin{align*}
        T(m,n) &\leq \bigO\left( \frac{m^{(|X|+1)/2}}{n} \left(m+m^{(|Y|-2+\omega)/2+o(1)}\right)\right)  & \\ &\leq 
        \frac{m^{\left(|X|+|Y| - 1 + \omega \right)/2+o(1)}}{n} & (|Y|\geq 2) \\
        & =  \frac{m^{\left(k - 1 + \omega \right)/2+o(1)}}{n} & (|X|+|Y| = k)\\ 
        &= \left(t_P(n,m)\cdot m^{(\omega-2)/2}\right)^{1+o(1)} & \text{(Observation \ref{obs:tp-breakdown})}
    \end{align*}
    We note that for the remaining case when $S=\emptyset$ and the heavy vertex in the solution corresponds to a vertex in $Y$, the analysis is almost completely identical to the one above, except by Lemma \ref{lemma:enumerating-T}, we would have an additional factor of $\frac{n}{\sqrt{m}}$ for enumerating the subgraphs isomorphic to $P[X]$, and would have a factor of $\frac{\sqrt{m}}{n}$ for solving the Maximal $|Y|$-Independent Set. Clearly those two factors cancel out.
\end{proof}
Things get slightly more complicated when dealing with patterns that only have one isolated vertex. 
In particular, the preprocessing part of the algorithm above takes in the worst case $\Theta\left(t_{P[X]}(n,m)\cdot m\right)$, while for pattern graphs with only one isolated vertex, we have $t_P(n,m) = \bigO\left(t_{P[X]}(n,m)\cdot \sqrt m\right)$. 
To circumvent this overhead, we employ a more careful analysis and use a hashing-based approach similar to that in Section \ref{ssec:small-graphs-algorithms}.
We first provide an algorithm for two special cases of some of the smallest patterns that exhibit this structure and then the idea is to reduce any pattern to one of these special cases.
\begin{lemma}\label{lemma:isolated-plus-edge}
    Let $G$ be a graph with $n$ vertices and $m$ edges, and let $Y\subseteq V(G)$ be a subset of vertices of $G$.
    Then we can construct an algorithm that w.h.p. enumerates all choices of three vertices $u,v,w$ in time ${m^{(\omega+1)/2+o(1)}}$, such that the following is satisfied.
    \begin{itemize}
        \item The vertex $w$ is contained in $Y$ and satisfies $\deg(w)< \sqrt{m}$.
        \item Vertices $u$ and $v$ are adjacent in $G$ and both of them are non-adjacent to $w$ in $G$ (i.e. the induced subgraph $G[\{u,v,w\}]$ is isomorphic to an edge and an isolated vertex).
        \item $Y\subseteq N[u]\cup N[v]\cup N[w]$.
    \end{itemize}
\end{lemma}
\begin{proof}
    Without loss of generality, assume that $\deg(u)\geq \deg(v)$.
    Let $t$ be such that $t\leq \deg(u)\leq 2t$.
    Assume first that $t\leq \sqrt{m}$. 
    Then each vertex $u,v,w$ has degree bounded by $\bigO(\sqrt{m})$ and so, if there exists a valid triple that dominates $Y$, the size of $Y$ must also be bounded by $\bigO(\sqrt{m})$.
    Hence by Lemma \ref{lem:guessing-approach} we can enumerate all valid triples in time $\MM(m,\sqrt{m},\sqrt{m}) \leq m^{(\omega+1)/2+o(1)}$, as desired.\footnote{Since vertices $u,v$ are adjacent, there are at most $m$ choices for this pair, and $w$ is by assumption in $Y$, and $|Y|\leq \bigO(\sqrt{m})$}
    For the rest of this proof we assume that $t\geq \sqrt{m}$.
    We now proceed to enumerate all valid triples $u,v,w$ such that $\deg(v)\leq \sqrt{m}$. 
    To this end, we run the following routine. 
    For each of the $\bigO(m/t)$ valid choices for the vertex $u$, construct the tripartite graph 
    $H_u = (V_u^{(1)}\cup V_u^{(2)}\cup V_u^{(3)}, E_u)$, where $V_u^{(1)}$ corresponds to a copy of $N(u)$, while $V_u^{(2)}$ and $V_u^{(3)}$ each correspond to a copy of the set $Y- N(u)$. The set of edges $E_u$ is constructed naturally, that is, for each pair $x\in V_u^{(i)}, y\in V_u^{(j)}$ for $i\neq j$ we add an edge $xy$ if and only if the edge between the corresponding vertices is present in $G$ (we assume that each vertex is adjacent to itself in $G$). 
    It is easy to see that for any fixed $u$, any valid choice of $v,w$ corresponds precisely to a choice of $v\in V_u^{(1)}$ and $w\in V_u^{(2)}$ such that the following two conditions are satisfied: 1) $V_u^{(3)}\subseteq N(w)\cup N(v)$ and 2) $vw \not\in E(H_u)$.
    Now, since we are only enumerating triples $u,v,w$ where both $v,w$ have degrees at most $\sqrt m $, if for a fixed $u$ the set $V_u^{(3)}$ contains more than $2\sqrt{m}$ vertices, we can conclude that no pair $v,w$ satisfies both conditions and proceed with the next choice of $u$. 
    On the other hand, if $|V_u^{(2)}| = |V_u^{(3)}| \leq 2\sqrt{m}$, by Lemma \ref{lem:guessing-approach}, for any fixed $u$, we can enumerate all the valid pairs of non-adjacent vertices $v,w$ that dominate $Y-N(u)$ in time 
    \begin{align*}
        T_u(n,m) &\leq \MM\left(V_u^{(1)},V_u^{(3)},V_u^{(2)} \right) \\ & \leq \MM\left(t, \sqrt{m}, \sqrt{m}\right) \\
        & \leq \bigO\left(\frac{t}{\sqrt{m}}\right)\MM\left(\sqrt{m},\sqrt{m},\sqrt{m}\right) & (t\geq \sqrt{m})\\
        & \leq t \cdot m^{\frac{\omega-1}{2} + o(1)}.
    \end{align*}
    Repeating this for each of the possible $\bigO(m/t)$ many choices of vertex $u$ yields the desired time.
    It remains to enumerate the triples of vertices $u,v,w$ where $\deg(v)\geq \sqrt{m}$. 
    Recall also that $2t\geq \deg(u)\geq \deg(v)$ and $t\geq \sqrt{m}$.
    Hence, since by the first condition of the Lemma, we are only interested in the triples $u,v,w$ with $\deg(w)\leq \sqrt{m} \leq t$, if the set $Y$ contains more than $5t$ vertices, we can conclude that no triple $u,v,w$ satisfies $Y\subseteq N[u]\cup N[v] \cup N[w]$.
    Hence, we can assume that $Y \leq 5t$.
    Let $B$ be a matrix whose rows correspond to adjacent pairs of vertices $u,v$ with $t\leq \deg(u)\leq 2t$ and $\deg(v)\geq \sqrt{m}$ and whose columns correspond to the set $Y$.
    We set the entry $B[(u,v), y]$ to $0$ if either $u$ or $v$ are adjacent to $y$ and to $1$ otherwise.
    Similarly, we construct a matrix $C$ whose columns correspond to the set of vertices $w\in Y$ such that $\deg(w)\leq \sqrt{m}$ and the rows correspond to $Y$.
    Set the entry $B[y,w]$ to $1$ if $w$ and $y$ are adjacent (we understand that each vertex is adjacent to itself) and $0$ otherwise.
    Notice that 
    $u,v,w$ dominate $Y$ if and only if $(B\cdot C)[(u,v),w] = \sum_k B[w,k]$ (each vertex that is not dominated by the pair $(u,v)$ is dominated by $w$). 
    Also notice that the matrix $B$ is an $\bigO\left(  \frac{m\sqrt m}{t} \right)\times \bigO\left(t\right)$ and $C$ is an $\bigO\left(  t \right)\times \bigO\left(t\right)$ matrix with at most $m$ non-zeros. Hence by a Lemma from
    ~\cite{FischerKR24}, we can enumerate with high probability all valid triples $u,v,w$ in time $\sqrt{m}\cdot m^{\omega/2+o(1)} = m^{\frac{\omega+1}{2}+o(1)}$, as desired.
    We remark that we enumerate all valid triples for any fixed $t$ in the claimed time. To enumerate all the valid triples, simply run the previous algorithm $\bigO(\log (n))$ times, once for each $\ell \in 1,\dots, \lfloor\log(n)\rfloor$, setting $t = 2^\ell$.
\end{proof}
An almost identical argument as in previous lemma can also be used to show that we can efficiently enumerate the solutions that are isomorphic to $C_{2r+1}+K_1$. 
\begin{lemma}\label{lemma:isolated-plus-cycle}
    Let $G$ be a graph with $n$ vertices and $m$ edges, and let $Y\subseteq V(G)$ be a subset of vertices of $G$.
    Then we can construct an algorithm that w.h.p. enumerates all choices of $2r+2$ vertices $x_1,\dots, x_{2r+1}, w$ in time ${m^{r + \omega/2+o(1)}}$, such that the following is satisfied.
    \begin{itemize}
        \item The vertex $w$ is contained in $Y$ and satisfies $\deg(w)< \sqrt{m}$.
        \item Vertices $x_1,\dots, x_{2r+1}$ form a cycle in $G$ (not necessarily induced) and $w$ is not adjacent to any of the vertices $x_i$ in $G$.
        \item $Y\subseteq N[w]\cup N[x_1]\cup\dots\cup N[x_{2r+1}]$.
    \end{itemize}
\end{lemma}
The idea for any pattern with precisely one isolated vertex $v$ is to decompose the pattern into sets $S,N(S),R$ as in Definition \ref{def:parameter}, and then depending on whether the set $S$ is empty or not, we either directly reduce to the setting of Lemma \ref{lemma:isolated-plus-edge}, or first apply Lemma \ref{lem:decomp-T} to decompose $R$ into edges and odd cycles, and then reduce to either the setting of Lemma \ref{lemma:isolated-plus-edge}, or that of Lemma \ref{lemma:isolated-plus-cycle}.
\begin{restatable}{proposition}{lemmaOneIsolatedVertex}
    Let $P$ be a pattern with one isolated vertex. Then there exists a randomized algorithm solving $P$-Dominating Set in time $\bigO\left(t_P(n,m)\right)\cdot m^{\frac{\omega-2}{2} + o(1)}$.
\end{restatable}
\begin{proof}
    Let $x$ be the isolated vertex in $P$. It is easy to verify that we can find all the solutions $D$ such that the vertex in $D$ that corresponds to $x$ has degree at least $\sqrt{m}$ in the claimed time, by simply guessing this vertex and reducing to the pattern without isolated vertices.
    Hence for the rest of the proof, we assume that $\deg(x)\le \sqrt{m}$.
    Moreover, since $D$ forms a dominating set in $G$ of size $k$, we know that $D$ contains a heavy vertex $v_h$.
    Consider now the partition of $P$ into sets $S,N(S),R$, as in Definition \ref{def:parameter}. 
    Assume first that $S\neq \emptyset$. Then let $u,v$ be an edge with $u\in S, v\in N(S)$. 
    By \cite{NgoPRR18}, we can enumerate all subsets of $G$ that are isomorphic to $P-\{u,v,x\}$ in time $\bigO\left ( \frac{t_P(n,m)}{m^{3/2}}\right)$. 
    For each such enumerated subgraph $Q$, in $\bigO(m)$ time we can construct the set $Y$ of all vertices that are not dominated by any vertex in $Q$ and using the algorithm from Lemma \ref{lemma:isolated-plus-edge}, we can enumerate all possible choices of vertices $\{u,v,x\}$ such that $u,v,x$ dominates $Y$ and $u,v$ are adjacent to each other and non-adjacent to $x$ (w.h.p.) in time $m^{(\omega+1)/2 + o(1)}$. It only remains to check for each enumerated solution $Q\cup\{u,v,x\}$ that $G[Q\cup\{u,v,x\}]\cong P$ (which can be done in constant time). 
    In total this yields the running time bounded by 
    \[\bigO\left ( \frac{t_P(n,m)}{m^{3/2}} (m+m^{(\omega+1)/2+o(1)})\right) \leq \bigO\left ( {t_P(n,m)}\right)\cdot m^{(\omega-2)/2+o(1)}).\]

    Assume now that $S = \emptyset$, in particular $R = V(P)$.
    By Lemma \ref{lem:decomp-T}, we can cover the graph $P - v_h - x$ with $\alpha$ independent edges and $\beta$ odd cycles. 
    Assume first $\alpha>0$.
    Fix one of the $\alpha$ edges $uv$.
    We notice that we can enumerate all subgraphs of $G$ isomorphic to $(\alpha-1)$ independent edges in $\bigO\left(m^{\alpha-1}\right)$ and each of the $\beta_r$ cycles of length $(2r+1)$ in time $\bigO\left(m^{(2r+1)/2}\right)$. Moreover, finding a heavy vertex takes $\bigO\left(\frac m n\right)$.
    Hence, for any fixed edge we can enumerate all the subgraphs isomorphic to $P-\{u,v,x\}$ takes at most $\bigO\left(\frac{m^{(k-2)/2}}{n}\right)$.
    Furthermore, by Observation \ref{obs:tp-breakdown}, this is precisely $\bigO\left(\frac{t_P(n,m)}{m^{3/2}}\right)$.
    We then apply the Lemma \ref{lemma:isolated-plus-edge}, in the same way as we did for $S\neq \emptyset$ to obtain the time $\bigO\left ( {t_P(n,m)}\right)\cdot m^{(\omega-2)/2+o(1)})$.
    If, however $\alpha=0$, we then approach very similarly, except that we fix one of the $\beta$ odd cycles $C_{2r+1}$ that do not contain $v_h$ and enumerate $P-C_{2r+1}$ in time $\bigO\left( \frac{t_P(n,m)}{m^{(2r+1)/2}} \right)$.
    We then conclude similarly as above, but this time we apply Lemma \ref{lemma:isolated-plus-cycle} to conclude the proof.
\end{proof}

\section{Conditional lower bounds}\label{sec:lower-bounds}
This section contains the proof of Theorem \ref{thm:lower-bound}, which states that the time complexity $t_P(n, m)$ for $P$-Dominating Set is optimal up to subpolynomial factors unless OVH fails.
The proof for this is similar to the proof given in \cite{FischerKR24} for the general Sparse $k$-Dominating Set problem.
This proof is also already adapted for Clique-, Independent-, and Matching-Dominating Set in \cite{KuennemannR24}.
We reduce from Orthogonal Vector with $k$ sets.
This reduction builds a graph with $\tilde{\Theta}(n)$ vertices and $\tilde{\Theta}(m)$ edges with $m = n^\gamma$ for some $\gamma \in [1, 2]$ to prove the time complexity for a specified graph density.
As mentioned in section \ref{sec:preliminaries}, one may choose the size of the vector sets of the $k$-OV problem.

Let $S \subset V(P)$ be a set of vertices as defined for the parameter $\rho$ in definition \ref{def:parameter}.
If $S \neq \emptyset$, we reduce from $k$-OV with $|S|$ sets of size $n$, $|N(S)|$ sets of size $m / n$, and $k - (|S| + |N(S)|)$ sets of size $m^{1 / 2}$.
Otherwise, if $S = \emptyset$, we reduce from $k - 1$ sets of size $m^{1 / 2}$ and one set of size $m / n$.
An instance of $k$-OV with these set sizes is then transformed into a $P$-Dominating Set instance by the following steps:
\begin{enumerate}
    \item For each vector set $A_i$, add a vertex group $V_i = \{v_{i, 1}, \dots, v_{i, |A_i|}\}$.
    The idea is to have the selection of the vertex $v_{i, j}$ into the dominating set correspond to choosing the $j$-th vector from $A_i$.
    \item For each vertex group representing one $A_i$, add a vertex set $R_i$ with $\max(2, m / |A_i|)$ vertices.
    For each $i \in [k]$, connect all $v_{i, j} \in V_i$ with all vertices of $R_i$.
    Because every vertex of $R_i$ has to be dominated, one vertex has to be chosen out of the set $V_i \cup \{r_i\}$ for each $i \in [k]$.
    There is no intersection between these sets and the $k$-dominating set may only pick $k$ vertices,
    therefore,
    the algorithm has to pick one vertex from each of these sets.
    \item For the dimension $d$, add a vertex of group $X$ of size $d$.
    Connect $t \in X$ with $v_{i, j}$ if the vector represented by $v_{i, j}$ has a $0$-entry at the $t$-th position.
    \item For every pair of vertex sets $V_i$ and $V_j$, connect all vertices of $V_i$ with all vertices $V_j$ if $ij$ is an edge in the pattern graph $P$.
    If $V_i$ instead corresponds to an isolated vertex in $P$, connect the vertices of $V_i$ pairwise.
\end{enumerate}
\begin{figure}[ht]
	\centering
	\begin{tikzpicture}
        \draw[black, dashed]
            (2, 4.2) -- (8, 0)
            (2, 2.0) -- (8, 0)
            (2, -0.8) -- (8, 0);

        \draw[black, fill=white] (8, 0) ellipse (0.4 and 1.3);
        \draw[every node/.style={inner sep=0pt, minimum size=2mm, fill, circle}]
            node (d0) at (8, -1) {}
            node (d1) at (8, -0.5) {}
            node (d2) at (8, 0) {}
            node (d3) at (8, 0.5) {}
            node (d4) at (8, 1) {};

        \draw[black]
		(2, 4.2) to[out=180, in=130] (-0.5, 2)
		(2, 4.2) to[out=180, in=130] (-0.5, -0.8)
		(2, 4.2) to[out=180, in=130] (-0.5, -3.7);
        
        \draw[black, fill=white] (2, -3.5) ellipse (2.6 and 0.9);

        \coordinate (v11pos) at (0, -3.7);
        \coordinate (v12pos) at (1, -3.7);
        \coordinate (v13pos) at (2, -3.7);
        \coordinate (v1lpos) at (4, -3.7);
        
        \draw[black]
            (v11pos) -- (d2)
            (v11pos) -- (d3)
            (v12pos) -- (d0)
            (v13pos) -- (d1)
            (v13pos) -- (d4)
            (v1lpos) -- (d1)
            (v1lpos) -- (d4);

        \draw[every node/.style={inner sep=0pt, minimum size=2mm, fill, circle}]
            node[label={[rectangle, fill=white]90:$v_{1, 1}$}] (v11) at (v11pos) {}
            node[label={[rectangle, fill=white]90:$v_{1, 2}$}] (v12) at (v12pos) {}
            node[label={[rectangle, fill=white]90:$v_{1, 3}$}] (v13) at (v13pos) {}
            node[label={[rectangle, fill=white]90:$v_{1, |A_1|}$}] (v1l) at (v1lpos) {};
        \node at (3, -3.7) {$\dots$};

        \coordinate (r1) at (-2.2, -3.7);
        \draw[black]
            (v11) to[out=-110, in=-70] (r1)
            (v12) to[out=-120, in=-70] (r1)
            (v13) to[out=-130, in=-70] (r1)
            (v1l) to[out=-150, in=-70] (r1);
        \draw[black, fill=white] (r1) ellipse (0.5 and 0.3) node at (r1) {$R_1$};
        
        \coordinate (r2) at (-2.2, -0.8);
        \coordinate (rk1) at (-2.2, 2.0);
        \coordinate (rk) at (-2.2, 4.2);
        \draw[black]
            (-0.6, -0.8) -- (r2)
            (-0.6, 2.0) -- (rk1)
            (1.3, 4.2) -- (rk);
        \draw[black, fill=white]
            (r2) ellipse (0.5 and 0.3) node at (r2) {$R_2$}
            (rk1) ellipse (0.5 and 0.3) node at (rk1) {$R_{k - 1}$}
            (rk) ellipse (0.9 and 0.4) node at (rk) {$R_k$};
        
        \draw[black, fill=white] (2, -0.8) ellipse (2.6 and 0.9) node at (2, -0.8) {$V_2$};
        \draw[black] node at (2, 0.6) {$\vdots$};
        \draw[black, fill=white] (2, 2.0) ellipse (2.6 and 0.9) node at (2, 2.0) {$V_{k - 1}$};
        \draw[black, fill=white] (2, 4.2) ellipse (0.7 and 0.3) node at (2, 4.2) {$V_k$};

        \node at (4.5, -4.2) {$V_1$};
        \node at (8.5, -1.0) {$D$};
    \end{tikzpicture}
	\caption{
        An example for the reduction result when the pattern is a $K_{1, k - 1}$.
		The vertex group $V_1$ is shown in greater detail.
		A solid line between two groups of vertices is a biclique.
		The dashed lines between the dimension set and the vertex sets representing the vector sets indicate connections as described in step 3.
		Vertex group $V_k$ has potentially the same number of vertices as any of the other groups, but considerably less when the reduction shows the lower bound for the time complexities of $P$-Dominating Set on a sparse instance graph.
	}
	\label{fig:general-kov-kds-reduction}
\end{figure}
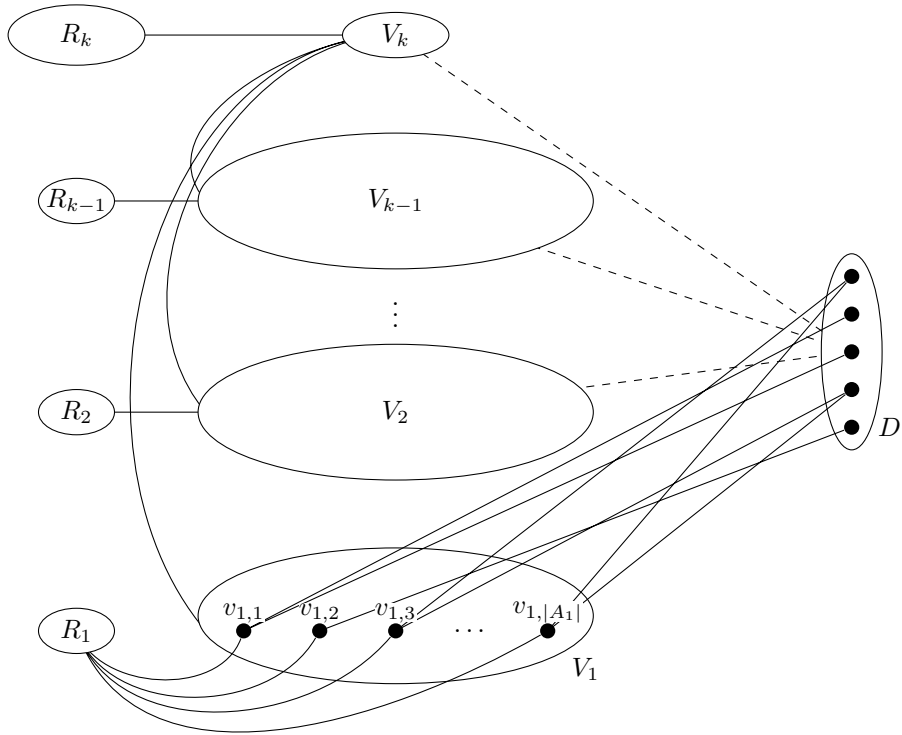

The correctness of this reduction of this reduction has multiple parts.
First, we show that $|V| \in \tilde{\Theta}(n)$ and $|E| \in \tilde{\Theta}(m)$ do in fact hold.
From this, it is also easy to see that the reduction runs in time $\tilde{\mathcal{O}}(m)$.
Next, we show that any valid dominating set $D$ in the graph $G$ constructed by the reduction induces $P$.
Last, we show that a dominating set $D$ exists if and only if the Orthogonal Vector instance has a solution.

To show that this reduction creates $\tilde{\Theta}(n)$ vertices, we analyze each step.
Step 1 creates at most $k \cdot n$ vertices.
Step 2 creates at most $k \cdot n$ vertices, and at least $n$ many due to the fact that at least one $A_i$ has the order $|A_i| = m / n$, which results in $|R_i| = n$.
Step 3 creates $d$ vertices, which is in $\tilde{\Theta}(1)$.
Step 4 creates no vertices.
The final number of vertices is in $\tilde{\Theta}(n)$, or more precisely bounded by:
\begin{equation}
    d + n \leq |V| \leq 2 \cdot k \cdot n + d
\end{equation}
To show that this reduction creates $\tilde{\Theta}(m)$ edges, we again analyze each step.
Step 1 creates no edges.
Step 2 creates at most $k m$ vertices, and at least $m$ many due to the fact that at least one $A_i$ has the order $|A_i| = m / n$, which results in $|R_i| = n$ and therefore $m$ edges in the pairwise connections between the vertices of $A_i$ and $R_i$.
Step 3 creates at most $k \cdot n \cdot d$ edges.
Step 4 creates at most $k \cdot (k - 1) / 2 \cdot m$ edges.
The final number of edges is in $\tilde{\Theta}(m)$, or more precisely bounded by:
\begin{equation}
    m \leq |E| \leq k \cdot (m + n \cdot d + (k - 1) / 2 \cdot m)
\end{equation}

Step 2 of the graph construction ensures that from each vertex set $A_i$, at least one vertex is taken.
Otherwise, the would be a vertex in some $R_i$ which is undominated.
Because $P$-Dominating Set implicitely limits the number of vertices which can be taken into the dominating set to $k$, there is exactly one vertex taken from each $A_i$ in a valid dominating set.
Further, because the vertices $A_i$ and $A_j$ are pairwise connected by an edge in step 4 if $ij$ is an edge in $P$, this dominating set will induce the pattern $P$.

Last is to argue that a solution to the Dominating Set problem on this graph exists if and only if there is a solution to the Orthogonal Vector instance.
Step 2 of the reduction already contains an argument that every solution to the Pattern Dominating Set problem takes exactly one vertex from each vertex set $V_i$.
If the vertex $v_{i, j}$ is chosen for the dominating set, choose the $j$th vector from $A_i$ in the Orthogonal Vector problem.
Because every vertex $t \in X$ is dominated, there exists one vertex $v$ such that $vt$ is an edge in the graph.
This edge exists exactly if the vector corresponding to $v$ has a $0$ for the $t$th dimension.
Therefore, for each dimension, there is now a vector which has $0$ for that dimension and we have a valid solution for the Orthogonal Vector problem.
Likewise, if vector $j$ from set $A_i$ is chosen in the solution to the Orthogonal Vector problem, then selecting the vertex $v_{i, j}$ into the dominating set gives a valid result.

\section{Generalization to pattern sets}\label{sec:pattern-sets}
A natural question following Pattern Dominating Set is to consider the case where the pattern $P$ does not have to be induced.
We take a step further:
Instead of a single pattern, we consider the problem with a given set of patterns $Q$, and the dominating set has to induce one of patterns in the graph.
This does in fact generalize the non-induced case with pattern $P$ by choosing the pattern set $Q$ as the set of all supergraphs of $P$ with the same order as $P$.
\begin{theorem}\label{thm:pattern-sets}
    Let $Q$ be a set of patterns with identical order.
    Then, the time complexity of finding a $Q$-Dominating Set is dominated by the pattern $P \in Q$ with the highest time complexity.
\end{theorem}

For the $Q$-Dominating Set problem, we consider $|Q|$ to be a constant which is part of the problem.
An algorithm which solves this problem can, for each $P \in Q$, solve the $P$-Dominating Set problem and answer yes if any of the patterns give yes as their answer.
The time complexity of this approach is
\begin{equation}
    \mathcal{O}(\sum_{P \in Q} t_P(n, m)) = \mathcal{O}(\max_{P \in Q}(t_P(n, m)))
\end{equation}

For the lower bound, apply the reduction from section \ref{sec:lower-bounds} to the pattern $\tilde{P} \in Q$ which maximizes $t_{\tilde{P}}(n, m)$.
This gives $Q$-Dominating Set the same lower bound as $\tilde{P}$-Dominating Set and therefore proves Theorem \ref{thm:pattern-sets}.
Because every dominating set of order $k$ in the graph built from this reduction induces a $\tilde{P}$, it is also intuitive that the $Q$-Dominating Set problem now has to do the work of finding a $\tilde{P}$-Dominating Set.
\bibliographystyle{plainurl}
\bibliography{references}
\appendix

\section{Computing the parameter $\rho$}
The parameter $\rho$ is a variation of the Fractional Independent Set.
For the original Fractional Independent Set problem, polynomial time algorithms are known using linear programming \cite{09-fractional-independent-set}.
It is also known that this linear program only produces values of $0$, $1 /2$ and $1$, which we can almost interpret as the vertex belonging to $N(S)$, $R$ and $S$, respectively.
The only thing we need to take care of is that parameter $\rho$ adds two additional constraints, which can be encoded in the linear program without breaking the property that the optimal solution produces only values in $\{0, 1 / 2, 1\}$.
First, definition \ref{def:parameter} requires that isolated vertices are not part of $S$, so it is not allowed to assign the value $1$ to isolated vertices.
This can simply be encoded by adding the condition $v_i \leq 1 / 2$ for every isolated vertex $v_i$.
Additionally, definition \ref{def:parameter} requires the set $S$ to be as large as possible without changing the value $|S| - |N(S)|$.
An equivalent requirement is to make $T$ as small as possible, which is to have the least number of vertices assigned the value $1 / 2$.
This requirement is well-known in solutions for linear relaxations of the Vertex Cover problem, see \cite{11-vertex-cover-LP} for an algorithm.

After this step, the sets $S$, $N(S)$ and $R$ are known.
This does not yet give the decomposition as described in Lemma \ref{lem:decomp-T}, but already has enough information to calculate the parameter $\rho$.
Recall that in definition \ref{def:parameter}, the parameter $\rho$ was defined as:
\begin{equation}
    \rho(P) = \begin{cases}
        -1 & S = \emptyset \\
        |S| - |N(S)| & S \neq \emptyset
    \end{cases}
\end{equation}

To find the decomposition of Lemma \ref{lem:decomp-T} once the vertices of the pattern are decomposed into $S$, $N(S)$ and $R$, we still need to determine the subset of edges that fulfill Lemma \ref{lem:decomp-T}.
Note that by a simple application of Hall's Marriage Theorem, there is a matching in the bipartite graph $(S \cup N(S), E_P \cap (S \times N(S)))$ covering all of $N(S)$.
This matching can be found using a matching algorithm on bipartite graph.
There are also algorithms to find a subset of edges in $R$ such that each vertex is part of an edge or an odd cycle, see for example \cite{10-edge-cycle-packing-algorithm}.

\end{document}